\documentclass[USenglish,oneside,twocolumn]{article}

\usepackage[utf8]{inputenc}
\usepackage[big]{dgruyter_NEW}
 
\cclogo{\includegraphics{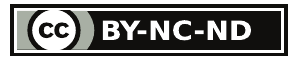}}

\newif\iffull

\fulltrue 

\usepackage{natbib}

\usepackage{amsmath}
\usepackage{amssymb, amsthm}

\usepackage{caption}
\usepackage{subcaption}

\usepackage{xcolor}
\definecolor{teal}{HTML}{008080}
\usepackage{hyperref}
\hypersetup{
   colorlinks=true,
   citecolor=teal
}

\usepackage{tikz}
\usetikzlibrary{arrows.meta}
\usetikzlibrary{calc, positioning, fit, shapes.misc}

\newtheorem{theorem}{Theorem}
\newtheorem{lemma}{Lemma}
\newtheorem{proposition}{Proposition}
\newtheorem{corollary}{Corollary}
\theoremstyle{definition}
\newtheorem{definition}{Definition}
\newtheorem{example}{Example}

\newcommand{\descr}[1]{\vspace{0.2cm} \noindent \textbf{\sffamily #1}}

\usepackage{algorithmic}
\usepackage[noend,lined,linesnumbered]{algorithm2e}

\usepackage{mathtools}

\begin{document}

\author[1]{Josh Smith}

\author*[2]{Hassan Jameel Asghar}

\author[3]{Gianpaolo Gioiosa}

\author[3]{Sirine Mrabet}

\author[4]{Serge Gaspers}

\author[3]{Paul Tyler}

\affil[1]{Work was done when Josh Smith was at Data61 (CSIRO), Australia, e-mail: j\_smith95@live.com}

\iffull
\affil[2]{Macquarie University and Data61 (CSIRO), Australia, e-mail: hassan.asghar@mq.edu.au. $^\dagger$This is the full version of the paper with the same title to appear in the proceedings on the 22nd Privacy Enhancing Technologies Symposium (PETS 2022).}
\fi

\affil[3]{Data61 (CSIRO), Australia, E-mail: \{gianpaolo.gioiosa, sirine.mrabet, paul.tyler\}@data61.csiro.au}

\affil[4]{University of New South Wales, Australia, e-mail: serge.gaspers@unsw.edu.au} 

\iffull
\title{Making the Most of Parallel Composition in Differential Privacy$^\dagger$}
\fi

\runningtitle{Making the Most of Parallel Composition in Differential Privacy}


\begin{abstract}{The parallel composition theorem of differential privacy states that the privacy loss incurred by a set of queries on disjoint partitions of the data domain is only the maximum loss incurred by the queries individually. Given a set of queries, this promises significant savings in utility, if it is possible to determine how many of the queries compose in parallel. We show that this `optimal' use of the parallel composition theorem corresponds to finding the size of the largest subset of queries that `overlap' on the data domain, a quantity we call the \emph{maximum overlap} of the queries. It has previously been shown that a certain instance of this problem, formulated in terms of determining the sensitivity of the queries, is NP-hard, but also that it is possible to use graph-theoretic algorithms, such as finding the maximum clique, to approximate query sensitivity. In this paper, we consider a significant generalization of the aforementioned instance which encompasses both a wider range of differentially private mechanisms and a broader class of queries. We show that for a particular class of predicate queries, determining if they are disjoint can be done in time polynomial in the number of attributes. For this class, we show that the maximum overlap problem remains NP-hard as a function of the number of queries. However, we show that efficient approximate solutions exist by relating maximum overlap to the clique and chromatic numbers of a certain graph determined by the queries. The link to chromatic number allows us to use more efficient approximate algorithms, which cannot be done for the clique number as it may underestimate the privacy budget. Our approach is defined in the general setting of $f$-differential privacy, which subsumes standard pure differential privacy and Gaussian differential privacy. We prove the parallel composition theorem for $f$-differential privacy. We evaluate our approach on synthetic and real-world data sets of queries. We show that the approach can scale to large domain sizes (up to $10^{20000}$), and that its application can reduce the noise added to query answers by up to 60\%.
}
\end{abstract}


   \journalname{Proceedings on Privacy Enhancing Technologies}
 \DOI{Editor to enter DOI}
   


\maketitle

\section{Introduction}


The sequential~\citep{dp-book} and parallel~\citep{mcsherry2009privacy} composition theorems of differential privacy are tools for understanding how privacy loss accumulates when a sensitive data set is queried multiple times. For example, if the \emph{individual} privacy losses for two given queries are $\epsilon_1$ and $\epsilon_2$, respectively, then according to the simplest sequential composition theorem, the \emph{combined} privacy loss for the queries is less than or equal to $\epsilon_1 + \epsilon_2$. On the other hand, if the two queries cover disjoint subsets of the data domain, then according to the parallel composition theorem, the combined privacy loss is just the maximum of $\epsilon_1$ and $\epsilon_2$.

There has been considerable research effort dedicated to the study of sequential composition in differential privacy (see, e.g., \citep{kairouz2015composition, mironov2017renyi, murtagh2016complexity, dong2019gaussian}). In particular, much work has been devoted to deriving sequential composition theorems that give tighter bounds on combined privacy loss than the bound provided by the simplest such theorem (see, e.g., \citep{kairouz2015composition}). On the other hand, the parallel composition theorem is usually not studied for its intrinsic interest, but rather merely applied as part of an analysis of a specific data release mechanism (see, e.g., \citep{diff-gen, asghar2019differentially}).

We target the following use case of practical importance. A data custodian allows users to issue queries on a sensitive data set answered via a differentially private mechanism. 
The data set is defined over a finite, discrete domain. 
The users require the custodian to measure the privacy loss as accurately as possible, so that given a fixed bound on the privacy loss either the number of queries that can be answered is maximized or the amount of noise added to query answers is minimized. To accurately measure privacy loss, the custodian should leverage parallel composition. Given a set of queries $Q$, this amounts to determining the largest subset of $Q$ such that all the queries in the subset `overlap' --- what we call the \emph{maximum overlap} of $Q$. The naive method of finding the maximum overlap begins by evaluating each query over all domain elements and comparing set intersections. This procedure is inefficient, as the queries need to be evaluated over the entire domain, requiring time exponential in the number of attributes in the domain. Moreover, even if there is an efficient way to determine the \emph{coverage} of a query over the domain, determining the maximum overlap via the naive method of running through all possible subsets of queries requires time exponential in the number of queries.

Computing the maximum overlap of a set of queries $Q$ is related to its $l_1$-sensitivity. The $l_1$-sensitivity of $Q$ is the maximum sum of absolute differences in the answers to queries from $Q$ taken over all neighboring data sets $D$ and $D'$ (differing in a single row) from some domain $\mathbb{D}$. Taking $Q$ as the set of \emph{statistical range queries}, i.e., conjunctions of range predicates on individual attributes, the authors from \cite{xiao-np} show that the problem of computing the $l_1$-sensitivity of $Q$ is NP-hard. In~\citep{inan-sensitivity}, by representing each statistical range query as a vertex, and introducing an edge between two vertices if the ranges of the corresponding queries overlap, the authors prove that the $l_1$-sensitivity of $Q$ is lower bounded by the cardinality of the maximum clique (the clique number) of the graph. While finding the clique number is still NP-hard, it is a well-studied problem and various efficient algorithms exist to (exactly) compute it in practice, implying that $l_1$-sensitivity can be well approximated. While these are important results, there are some key shortcomings which we seek to address in this paper.
\begin{itemize}
    \item The notion of (weighted) maximum overlap\footnote{Weighted maximum overlap is the variant of maximum overlap that treats the case in which different queries may be allocated different privacy budgets. See Section~\ref{sub:max-overlap} for the exact definition.} (introduced in this paper) of a set of queries $Q$, and hence the optimal use of parallel composition, is a more general problem than finding the $l_1$-sensitivity of $Q$. For composition under $\epsilon$-differential privacy, one can show that computing the weighted maximum overlap is equivalent to finding $l_1$-sensitivity. However, $l_1$-sensitivity analysis excludes several prominent mechanisms (and hence heterogeneous composition involving such mechanisms) such as the Gaussian mechanism, which is proven differentially private in the \emph{approximate} sense~\citep{dp-book} or under concentrated differential privacy~\citep{dwork-cdp, bun-cdp} using $l_2$-sensitivity of queries. Since the $l_2$-sensitivity of $Q$ is less than or equal to its $l_1$-sensitivity, the upper bound on $l_1$-sensitivity may be loose. Focusing on the problem from the perspective of optimal parallel composition decouples it from the underlying sensitivity metric, so that solutions are applicable to other, relaxed notions of differential privacy.
    \item To accomplish this, we characterize the notion of maximum overlap in terms of $f$-differential privacy \citep{dong2019gaussian}, which is a recently proposed framework that subsumes all notions of differential privacy that admit an interpretation in terms of hypothesis testing. To this end, we prove a parallel composition theorem and an optimal composition theorem for $f$-differential privacy. This enables us to illustrate our approach in both the well-established setting of pure differential privacy~\citep{calib-noise, dp-book} and the new setting of Gaussian differential privacy \citep{dong2019gaussian}.
    \item To find instances of the targeted use case which can be solved efficiently in the number of attributes and the number of queries, we restrict our focus to queries that are conjunctions of arbitrary predicates on individual attributes. These queries properly subsume the statistical range queries studied in~\citep{xiao-np, inan-sensitivity}, and have been studied in the differential privacy literature (see, e.g., \citep{mckenna}). We show that, with such queries, one can check whether a given subset of queries has non-empty intersection in time polynomial in the number of attributes. We further show that the `optimal' parallel composition of such queries can be profitably analyzed using an intersection graph, as was done in~\citep{inan-sensitivity} for the computation of $l_1$-sensitivity.
    \item Representing the problem as an intersection graph allows us to use algorithms for computing the clique number to approximate the maximum overlap, similar to~\citep{inan-sensitivity}. However, in~\citep{inan-sensitivity} and also in our case, we are bound to use `exact' algorithms for the computation of the clique number, as any approximate clique number may underestimate the total privacy loss. This is a drawback, since approximate algorithms are often more efficient than exact algorithms. We further upper bound the maximum overlap problem by the chromatic number of the graph, and since any approximate chromatic number is always greater than or equal to the exact chromatic number, it never underestimates the privacy loss. This allows us to use algorithms for computation of the approximate chromatic number, which run for larger sets of queries and larger domain sizes than exact algorithms for computation of the clique number.
    \item We evaluate exact clique number and approximate chromatic number algorithms proposed in the literature by varying the domain size and number of queries, and show that the latter can be used to approximate the maximum overlap (and hence optimal parallel composition) for much larger sets of queries and domains. For instance, it can handle more than 1,000 queries for domains of size up to $10^{2500}$ (Section~\ref{subsec:scalingexperiments}). Through experiments on synthetic and real census query data sets, we show that there is likely to be significant overlap between queries, and hence using our approach results in significant gain in utility, resulting in noise reduction of up to 95\% for synthetic queries (Section~\ref{subsec:randomcensusqueries}) and up to 58.5\% (36.2\% on average) for real census queries
    (Section~\ref{sub:absqueryresults}). 
\end{itemize}

\section{\texorpdfstring{$f$}{f}-Differential Privacy and Optimal Parallel Composition}

\subsection{Preliminaries}

An \emph{attribute} $A$ is a finite set, whose elements are called attribute \emph{values}. A \emph{domain}, denoted $\mathbb{D}$, is the Cartesian product of $m \ge 1$ attributes: $\mathbb{D} \coloneqq A_1 \times \cdots \times A_m$. An element of a domain is called a \emph{row}. A data set $D$ is a subset of $\mathbb{N}^{|\mathbb{D}|}$, in the histogram notation~\citep{dp-book}. 
Let $\mathbb{D}'$ be a subset of $\mathbb{D}$. The intersection $D \cap \mathbb{D}'$ is the set of rows of the data set $D$ that are in $\mathbb{D}'$. Two data sets $D, D'$ on $\mathbb{D}$ are \emph{neighboring}, denoted $D \sim D'$, if they differ in a single row.





\subsection{Standard Differential Privacy: Background}

\begin{definition}[Differential privacy~\citep{calib-noise, dp-book}]
A mechanism (randomized algorithm) $M$ is \emph{$(\epsilon, \delta)$-differentially private} if for all $S \subseteq R$, where $R$ is the outcome space of $M$, and all neighboring data sets $D \sim D'$, one has $\mathrm{Pr}(M(D) \in S) \le e^{\epsilon}\:\mathrm{Pr}(M(D') \in S) + \delta$, where $\epsilon$ and $\delta$ are non-negative real numbers. If $\delta = 0$, one says that $M$ is \textit{$\epsilon$-differentially private}.
\end{definition}


When $\delta = 0$, the resulting notion is sometimes called \emph{pure} differential privacy, in contrast to the notion of \emph{approximate} differential privacy for $\delta > 0$. An important property of differential privacy is that it composes~\citep{dp-book}.

\begin{theorem}[Sequential composition]
\label{the:basic-comp}
Let $M_1, \ldots, M_k$ be a sequence of mechanisms. If all the mechanisms in the sequence are $(\epsilon, \delta)$-differentially private, then the composition of the sequence is $(k\epsilon, k\delta)$-differentially private.\qed
\end{theorem}



The above result is a simple example of a \emph{sequential} composition theorem for differential privacy, as opposed to a \textit{parallel} composition theorem, which is given next.

\begin{theorem}[Heterogeneous parallel composition~\citep{mcsherry2009privacy}]
\label{the:par-comp}
Let $\mathbb{D}$ be a domain, let $D \in \mathbb{N}^{|\mathbb{D}|}$ be a data set, and let $k$ be a positive integer. For each $i \in [k]$, let $\mathbb{D}_i$ be a subset of $\mathbb{D}$, let $M_i$ be a mechanism that takes $D \cap \mathbb{D}_i$ as input, and suppose $M_i$ is $\epsilon_i$-differentially private. If $\mathbb{D}_i \cap \mathbb{D}_j = \emptyset$ whenever $i \neq j$, then the composition of the sequence $M_1, \ldots, M_k$ is $\max\,\{\epsilon_i : i \in [k]\}$-differentially private.\qed
\end{theorem}
Although it is not explicitly stated in the theorem, it should be noted that each member of the sequence of mechanisms may also take the outputs of its predecessors it as input.

\begin{definition}[Laplace mechanism~\citep{calib-noise}]
\label{def:laplace}
The zero-mean Laplace distribution has the probability density function $\text{Lap}(x \mid b) \coloneqq \frac{1}{2b} e^{-\frac{|x|}{b}}$, where $b$, a non-negative real number, is a scale parameter. Let $Q$ be a set of queries each of which maps data sets to real numbers. The $l_1$-sensitivity of $Q$, denoted $\Delta Q$, is defined as 
\[
\Delta Q = \max_{\substack{D, D'\\ D \sim D'}} \lVert Q(D) - Q(D') \rVert_1
\]
Given a data set $D$ and a set of $t$ queries $Q$, the Laplace mechanism is defined as $M(Q, D) \coloneqq Q(D) + (Y_1, \ldots, Y_t)$, where $Y_i$ is a Laplace random variable of scale $\Delta Q / \epsilon$. The Laplace mechanism is $\epsilon$-differentially private~\citep{calib-noise}. 
\end{definition}


\subsection{\texorpdfstring{$f$}{f}-Differential Privacy: Background}
It is well known that the sequential composition result in Theorem~\ref{the:basic-comp} for $(\epsilon, \delta)$-differential privacy is not tight. There have been a number of successful attempts to obtain tighter composition results by adopting relaxed notions of differential privacy, including the advanced composition theorem for approximate differential privacy~\citep{dp-book}, as well as concentrated differential privacy~\citep{dwork-cdp, bun-cdp} and R\'enyi differential privacy~\citep{mironov2017renyi}. In this paper, we focus on the notion of \emph{$f$-differential privacy} recently proposed by
Dong, Roth and Su~\citep{dong2019gaussian}, which is a generalization of standard differential privacy (i.e., of $(\epsilon, \delta)$-differential privacy) based on its hypothesis testing interpretation. Let $D$ and $D'$ be neighboring data sets given as input to a mechanism $M$. Given the output of the mechanism, the goal is to distinguish between two competing hypotheses: the underlying data set being $D$ or $D'$. Let $P$ and $P'$ denote the probability distributions of $M(D)$ and $M(D')$, respectively. Given any rejection rule $0 \le \phi \le 1$, the type-I and type-II errors are defined as follows ~\citep{dong2019gaussian}: $\alpha_\phi \coloneqq \mathbb{E}_P[\phi]$ and  $\beta_\phi \coloneqq 1-  \mathbb{E}_{P'}[\phi]$.

\begin{definition}[Trade-off function~\citep{dong2019gaussian}]
\label{def:T}
For any two probability distributions $P$ and $P'$ on the same
space, the \emph{trade-off function} $T(P, P') : [0, 1] \to [0, 1]$ is defined by
\[
T(P, P')(\alpha) \coloneqq \inf \{\beta_\phi : \alpha_\phi \leq \alpha\}
\]

\noindent for all $\alpha \in [0, 1]$, where the infimum is taken over all (measurable) rejection rules.
\end{definition}

A trade-off function gives the minimum achievable type-II error at any given level of type-I error. For a function to be a trade-off function, it must satisfy the conditions specified in the following proposition.

\begin{proposition}[\citep{dong2019gaussian}]
A function $f : [0, 1] \to [0, 1]$ is a trade-off function if and only if $f$ is convex, continuous and non-increasing, and $f(x) \leq 1 - x$ for all $x \in [0, 1]$.
\end{proposition}

Abusing notation, let $M(D)$ denote the distribution of a mechanism $M$ when given a data set $D$ as input.
\begin{definition}[$f$-differential privacy~\citep{dong2019gaussian}]
Let $f$ be a trade-off function. A mechanism $M$ is said to
be \emph{$f$-differentially private} if $T(M(D), M(D')) \geq f$ for all neighboring data sets $D$ and $D'$.
\end{definition}

\begin{definition}[Gaussian Differential Privacy]
\label{def:gauss}
A key example of $f$-differential privacy is \emph{Gaussian differential privacy}~\citep{dong2019gaussian}, which is based on the trade-off function
\[
G_\mu \coloneqq T(\mathcal{N}(0, 1), \mathcal{N}(\mu, 1)),
\]
where $\mu \geq 0$. This trade-off function can be written explicitly as $G_\mu \coloneqq \Phi(\Phi^{-1} (1 - \alpha) - \mu)$, where $\Phi$ is the standard normal CDF. An example of a $G_{\mu}$-differentially private mechanism (or $\mu$-GDP mechanism, for short) is the \emph{Gaussian mechanism}:
$M(q, D) \coloneqq q(D) + Y$, where $q$ is a query of sensitivity $\Delta q$ and $Y \sim \mathcal{N}(0, \Delta q^2/\mu^2)$~\citep{dong2019gaussian}. 
\end{definition}

\begin{definition}[Tensor product~\citep{dong2019gaussian}]
Let $P_1, P_2, P_3$ and $P_4$ be probability distributions. Let $f$ and $g$ be the trade-off functions $T(P_1, P_2)$ and $T(P_3, P_4)$, respectively. The \emph{tensor product} of $f$ and $g$, which is denoted $f \otimes g$, is defined by
\[
f \otimes g \coloneqq T(P_1 \times P_3, P_2 \times P_4).
\]
\end{definition}

This extends to $n$-fold tensor products due to the associativity of the tensor product~\citep{dong2019gaussian}. The following theorem is the basic sequential composition result for $f$-differential privacy.

\begin{theorem}[Sequential composition~\citep{dong2019gaussian}]
Let $M_i(\cdot, y_1, \ldots , y_{i-1})$ be $f_i$-DP for all $y_1 \in Y_1, \ldots , y_{i-1} \in Y_{i-1}$. Then the $n$-fold composed mechanism $M : X \to Y_1 \times \cdots \times Y_n$ is $f_1 \otimes \cdots \otimes f_n$-differentially private.
\end{theorem}

A corollary of the above is that the $n$-fold (sequential) composition of $\mu_i$-GDP mechanisms is $\sqrt{\mu_1^2 + \cdots + \mu_n^2}$-GDP~\citep{dong2019gaussian}. A mechanism $M$ is $(\epsilon, \delta)$-DP if and only if it is $f_{\epsilon, \delta}$-DP~\citep{stat-fw-dp, dong2019gaussian}, where $f_{\epsilon, \delta}$ is the trade-off function $
\max \{ 0, 1 - \delta - e^\epsilon \alpha, e^{-\epsilon}(1 - \delta - \alpha) \}$. 
\iffull
Finally, the following ties the composition of $(\epsilon, \delta)$-DP mechanisms to $f$-DP:
\begin{theorem}[Central limit theorem \citep{dong2019gaussian}]
\label{the:clt}
Let the privacy parameters of a sequence of $n$ $(\epsilon, \delta)$-DP mechanisms be arranged in a triangular array such that $\{(\epsilon_{ni}, \delta_{ni}) : 1 \leq i \leq n \}$. Assume:
\begin{equation*}
\begin{aligned}
\sum_{i=1}^{n} \epsilon_{ni}^2 \rightarrow \mu^2,  \max_{1 \leq i \leq n} \epsilon_{ni} \rightarrow 0, \; 
\sum_{i=1}^{n} \delta_{ni} \rightarrow \mu^2,  \max_{1 \leq i \leq n} \delta_{ni} \rightarrow \delta
\end{aligned}
\end{equation*}
for some non-negative constants $\mu, \delta$ as $n \rightarrow \infty$. Then, we have
\[
f_{\epsilon_{n1}, \delta_{n1}} \otimes \ldots \otimes f_{\epsilon_{nn}, \delta_{nn}} \rightarrow G_{\mu} \otimes f_{0, 1 - e^{-\delta}}
\]
uniformly over $[0,1]$ as $n \rightarrow \infty$.
\end{theorem}

For $(\epsilon, 0)$-DP, the authors of \citep{dong2019gaussian} show that the error of this approximation is bounded by $\mathcal{O}(1/n)$. 
\fi

\subsection{\texorpdfstring{$f$}{f}-Differential Privacy and Composition}
\label{sub:f-dp-comp}

In this section, we prove a parallel composition theorem for $f$-DP as a counterpart to Theorem~\ref{the:par-comp}, and then the optimal composition, in terms of the number of invocations of parallel and sequential compositions, of an arbitrary sequence of $f$-DP mechanisms. 
To prove these results we recall the notion of lower convex envelope (see, e.g., ~\citep[\S 2.4.2.3]{lce-phd}).
\begin{definition}
Let $f_1$ and $f_2$ be trade-off functions. The lower convex envelope $\breve{f}: [0, 1] \rightarrow [0, 1]$ of $f_1$ and $f_2$, denoted $\text{lce}\{f_1, f_2\}$, is defined as
\[
\breve{f}(x) \coloneqq \sup\,\{ f(x) \mid f \text{ is convex and } f \le \min\{f_1, f_2\} \}.
\]
\end{definition}
\begin{lemma}
\label{lem:lce}
 The lower convex envelope $\breve{f}$ of two trade-off functions $f_1$ and $f_2$ is a trade-off function.
\end{lemma}
\begin{proof}
\iffull
By definition $\breve{f}$ is convex. It is also non-increasing since it is less than or equal to $\min\{f_1, f_2\}$, both of which are non-increasing. Also, by definition, $\breve{f}(x)\le \min\{f_1, f_2\} \le 1 - x$ for all $x \in [0, 1]$. Since $\breve{f}$ is convex, it is continuous over $(0, 1)$. Since $f_1(1) = f_2(1) = 1$, we have $\breve{f}(1) = 1$. Then, in the half neighborhood of $(1, \breve{f}(1))$, the graph of $\breve{f}$ coincides with that of $f_1$ or $f_2$ or both~\citep[Theorem 2.5]{lce-phd}. Therefore, $\breve{f}$ is continuous at $1$, due to the continuity of both $f_1$ and $f_2$ at 1. At $x = 0$, if $f_1(0) = f_2(0)$, then the continuity of $\breve{f}$ follows due to a similar argument as above. So let us assume that is not the case, and without loss of generality, let $f_1(0) < f_2(0)$. Then the graph of $\breve{f}$ in the half neighborhood of $(0, \breve{f}(0))$ is either a straight line~\citep[Theorem 2.5]{lce-phd}, or coincides with that of $f_1$. In either case, it is continuous at 0. It follows that $\breve{f}$ is a trade-off function. 
\fi
\end{proof}

\begin{corollary}
\label{cor:lce}
 Let $\breve{f}$ be the lower convex envelope of two trade-off functions $f_1$ and $f_2$ such that $f_1 \le f_2$ over $[0, 1]$. Then $\breve{f} = f_1$.
\end{corollary}

Let $D$ and $D'$ be any two neighboring data sets from $\mathbb{N}^{|\mathbb{D}|}$. Let $\mathbb{D}_1$ and $\mathbb{D}_2$ be disjoint subsets of $\mathbb{D}$. Write $D_1 = D \cap \mathbb{D}_1$ and $D_2 = D \cap \mathbb{D}_2$. Analogously define $D'_1$ and $D'_2$. 
Let $T$ be the trade-off function as defined in Definition~\ref{def:T}.

\begin{theorem}[Parallel composition in $f$-DP]
\label{the:par-comp:gdp}
Let $M_1$ and $M_2$ be $f_1$-DP and $f_2$-DP mechanisms, respectively. The joint mechanism $M$ defined by $
M(D) \coloneqq (y_1, M_2(y_1, D_2))$, where $y_1 \coloneqq M_1(D_1)$, is $\text{lce}\{f_1, f_2\}$-DP.
\end{theorem}
\begin{proof}
We have
\begin{align}
    & T(M(D), M(D')) \nonumber\\
    &=T(M_1(D_1) \times M_2(y_1, D_2), M_1(D'_1) \times M_2(y_1, D'_2)) \nonumber\\
    &= T(M_1(D_1), M_1(D'_1)) \otimes T(M_2(y_1, D_2), M_2(y_1, D'_2))\text{.}\label{eq:part}
\end{align}
Since $D \sim D'$, either $D_1 \sim D'_1$ or $D_2 \sim D'_2$, but not both. Assume $D_1 \sim D'_1$. Then $D_2 = D'_2$, and Eq.~\ref{eq:part} becomes
\begin{align}
    &T(M(D), M(D')) \nonumber\\
    &= T(M_1(D_1), M_1(D'_1)) \otimes T(M_2(y_1, D_2), M_2(y_1, D_2)) \nonumber\\
                   &= T(M_1(D_1), M_1(D'_1)) \otimes \text{Id} \nonumber\\
                   &= T(M_1(D_1), M_1(D'_1)) \nonumber\\
                   &\geq f_1 \label{eq:f1},
\end{align}
where $\text{Id}$ is the trade-off function of two identical distributions, and the third step follows from the properties of the tensor product of trade-off functions~\citep[
\S 3.1]{dong2019gaussian}. Next assume $D_2 \sim D'_2$, which means $D_1 = D_1'$, and in this case Eq.~\ref{eq:part} becomes
\begin{align}
    & T(M(D), M(D')) \nonumber\\
    &= T(M_1(D_1), M_1(D_1)) \otimes T(M_2(y_1, D_2), M_2(y_1, D'_2))  \nonumber\\
                   &= \text{Id} \otimes T(M_2(y_1, D_2), M_2(y_1, D'_2)) \nonumber\\
                   &= T(M_2(y_1, D_2), M_2(y_1, D'_2)) \nonumber\\
                   &\geq f_2 \label{eq:f2}.
\end{align}
Combining Eqs.~\ref{eq:f1} and \ref{eq:f2}, for the unconditional distributions $M(D)$ and $M(D')$, we get

\[
T(M(D), M(D')) \ge \min\{f_1, f_2\} \ge \text{lce}\{f_1, f_2\}.
\]

\end{proof}
The above extends to any countable number of disjoint subsets of the data domain. In particular, for $k \ge 2$, we have that $M$ is $\text{lce}\{f_1, f_2, \ldots, f_k\}$-DP.

\begin{corollary}
\label{cor:parallelcompositiongdp}
Let a sequence of $k$ mechanisms $M_i$ each be $\mu_i$-GDP. Let $\mathbb{D}_i$ be disjoint subsets of $\mathbb{D}$. The joint mechanism defined as the sequence of $M_i(D \cap \mathbb{D}_i)$ (given also the output of the previous $i-1$ mechanisms) is $\max\{\mu_1, \mu_2, \ldots, \mu_k\}$-GDP.
\end{corollary}
\begin{proof}
\iffull
From Theorem~\ref{the:par-comp:gdp}, $M$ is $\text{lce}\{G_{\mu_1}, G_{\mu_2}, \ldots, G_{\mu_k}\}$-DP. From the definition of $G_\mu$~\citep{dong2019gaussian}, $G_{\mu} = \Phi(\Phi^{-1}(1 - \alpha) - \mu)$, where $\Phi$ is the standard normal CDF, $\mu \ge 0$, and $0 \leq \alpha \leq 1$. Fix any $\mu_i$ and $ \mu_j$ such that $\mu_i \neq \mu_j$. Equating $G_{\mu_i}$ and $G_{\mu_j}$, and noting that $\Phi$ is a strictly increasing function, we get
\[
\Phi^{-1}(1 - \alpha) - \mu_i = \Phi^{-1}(1 - \alpha) - \mu_j,
\]
which implies $\mu_i = \mu_j$, a contradiction. Thus, $G_{\mu_i}$ and $G_{\mu_j}$ do not intersect for all real numbers in $[0, 1]$. Assume that $G_{\mu_i} < G_{\mu_j}$. From Corollary~\ref{cor:lce}, $\text{lce}\{G_{\mu_i}, G_{\mu_j}\} = G_{\mu_i}$, and
\[
\Phi^{-1}(1 - \alpha) - \mu_i < \Phi^{-1}(1 - \alpha) - \mu_j,
\]
implies that $\mu_i > \mu_j$. The result follows. 
\fi
\end{proof}

The parallel composition theorem considers mechanisms that operate on disjoint subsets of the domain. However, we are interested in the more general case where mechanisms operate on arbitrary subsets of the domain. To address this, we introduce the concept of maximum overlap, which we explore in depth in subsequent sections.

\descr{Maximum overlap.} Let $\mathcal{M}_i$ be a sequence of $k$ mechanisms, each providing $f_i$-differential privacy. Let $\mathbb{D}_i$ be arbitrary subsets of the domain $\mathbb{D}$. Let $D \cap \mathbb{D}_i$ denote the input to the mechanism $M_i$, where $D$ is a data set. The mechanism is also given as input the outputs of the previous $i-1$ mechanisms. The maximum overlap $f_{\gamma}$ for the sequence of mechanisms is defined by
\[
f_{\gamma} \coloneqq \underset{I \subseteq \{1, \ldots, k\}}{\text{lce}}\left\{ \bigotimes_{i \in I} f_i  :  \bigcap_{i \in I} \mathbb{D}_{i} \neq \emptyset \right\}.
\]

The name `maximum' may be a bit confusing given the lower convex envelope and its relation to the minimum of the trade-off functions in the definition. This is because $f_{\gamma}$ is a trade-off function, and for specific definitions of privacy (e.g., $\epsilon$-DP or $\mu$-GDP), the minimization of the trade-off function corresponds to a maximization of the parameters (e.g., $\epsilon$ or $\mu$).

For Gaussian differential privacy, where each mechanism provides $\mu_i$-GDP, we can exactly characterize the maximum overlap as $G_{\gamma}$, where
\[
\gamma \coloneqq \max_{I \subseteq \{1, \ldots, k\}} \left\{ \sqrt{\sum_{i \in I} \mu_i^2}  :  \bigcap_{i \in I} \mathbb{D}_{i} \neq \emptyset \right\}.
\]

For $\epsilon$-differential privacy, where each mechanism provides $\epsilon_i$-DP, 
\iffull
we need to invoke the central limit theorem (\ref{the:clt}). It follows that $f_{\gamma} \rightarrow G_{\gamma}$, where
\[
\gamma \coloneqq \max_{I \subseteq \{1, \ldots, k\}} \left\{ \sqrt{\sum_{i \in I} \epsilon_i^2}  :  \bigcap_{i \in I} \mathbb{D}_{i} \neq \emptyset \right\}.
\]

The above serves as an approximation. We can also give a (loose) lower bound using the simple sequential composition theorem for $\epsilon$-DP. In this case 
\fi
we have that $f_{\gamma} \geq f_{\epsilon', 0}$, where $f_{\epsilon', 0}$ is the trade-off function of an $\epsilon'$-DP mechanism and
\[
\epsilon' \coloneqq \max_{I \subseteq \{1, \ldots, k\}} \left\{ \sum_{i \in I} \epsilon_i  :  \bigcap_{i \in I} \mathbb{D}_{i} \neq \emptyset \right\}.
\]

Or, in the special case where all mechanisms provide $\epsilon$-DP, we have

\[
\epsilon' \coloneqq \epsilon \max_{I \subseteq \{1, \ldots, k\}} \left\{ |I|  :  \bigcap_{i \in I} \mathbb{D}_{i} \neq \emptyset \right\}.
\]


This definition of maximum overlap leads to our theorem for composition of mechanisms operating on arbitrary subsets of the domain.

\begin{theorem}[Composition of arbitrary mechanisms]
\label{thm:arbitrarycomposition}
Let $\mathcal{M} \coloneqq \{M_i(D \cap \mathbb{D}_i)\}$, for $1 \le i \le k$, be a set of mechanisms, where $\mathbb{D}_i$ are subsets of the domain $\mathbb{D}$. Suppose that $M_i$ is $f_i$-differentially private. Then the composition of $\mathcal{M}$ is $f_{\gamma}$-differentially private, where $f_\gamma$ is the maximum overlap of $\{f_i \colon i \in [k]\}$.
\end{theorem}
\begin{proof}
Consider $I_1 \coloneqq \left\{I \subseteq [k] : \bigcap_{i \in I} \mathbb{D}_i \neq \emptyset\right\}$, 
and let $F_1 \coloneqq \left\{ \bigotimes_{i \in I} f_i : I \in I_1 \right\}$. Also, let $I_2 \coloneqq \left\{ I \in I_1 : \text{ for all } I' \in I_1, I \not\subset I' \right\}$, i.e., the set of all elements of $I_1$ which are not proper subsets of any other element in $I_1$. Finally, let $F_2 \coloneqq \left\{ \bigotimes_{i \in I} f_i : I \in I_2 \right\}$. We claim that $\min F_1 = \min F_2$. Since $I_2$ is a subset of $I_1$, we immediately have that $\min F_1 \leq \min F_2$. Next, consider $\min F_2$. Let $I' \in I_1$, and let $I'' \supseteq I'$ be a set (which is guaranteed to be in $I_2$ by construction). We see that
\begin{equation*}
\min F_2 \leq \bigotimes_{i \in I''} f_i = \bigotimes_{i \in I'} f_i \bigotimes_{i \notin I'} f_i
                          \leq \bigotimes_{i \in I'} f_i \bigotimes_{i \notin I'} \text{Id} = \bigotimes_{i \in I'} f_i,
\end{equation*}
where $\text{Id}$ is the trade-off function of two identical distributions~\citep{dong2019gaussian}. Above, we have used the fact that $\text{Id} \geq f$ for all trade-off functions $f$, and other properties of the tensor product~\citep[Section 3.1]{dong2019gaussian}. Therefore, $\min F_2 \leq \min F_1$. Thus, $\min F_1 = \min F_2$. Next, we claim that for all $I', I'' \in I_2$, the intersected domains $\bigcap_{i \in I'} \mathbb{D}_i$ and $\bigcap_{i \in I''} \mathbb{D}_i$ are disjoint. Assume to the contrary that they are not. Then $\left(\bigcap_{i \in I'} \mathbb{D}_i\right) \cap  \left(\bigcap_{i \in I''} \mathbb{D}_i\right) = \bigcap_{i \in I' \cup I''} \mathbb{D}_i \neq \emptyset $. This implies that $I' \cup I'' \in I_2$, a contradiction. Thus, the set of mechanisms $\{M_i(D \cap \mathbb{D}_i)\}$ is $\text{lce}\{ F_2\}$-DP according to Theorem~\ref{the:par-comp:gdp}. Since, $\min F_2 = \min F_1$, this is exactly the maximum overlap of $\{f_i : i \in [k]\}$.
\end{proof}

We note that if no subsets of the domain are disjoint, $f_{\gamma}$ is exactly the sequential composition of all mechanisms, and if all subsets of the domain are disjoint, $f_{\gamma}$ is given exactly the parallel composition of all mechanisms (Theorem~\ref{the:par-comp:gdp}).

\section{Predicate Queries and Maximum Overlap}
\label{sec:preliminaries}

In Section~\ref{sub:f-dp-comp}, we defined maximum overlap in terms of $f$-differential privacy. According to the definition, maximum overlap is determined by identifying mechanisms whose sub-domains overlap. In general, there are different $f$-differentially private mechanisms answering different types of queries, e.g., predicate and sum queries. In practice, however, there is often a single fixed mechanism, e.g., the Gaussian mechanism (Definition~\ref{def:gauss}), and a single class of queries, e.g., the predicate queries. In this case, one can determine maximum overlap using only information about the queries, i.e., by checking the subsets of the domain \emph{covered} by the queries.

In this section, we will show how maximum overlap relates to a given set $Q$ of $t$ queries. The data custodian could optimize the overall privacy budget usage using Theorem~\ref{thm:arbitrarycomposition}. Unfortunately, this 
procedure is exponential in $m$ (the number of attributes in the domain) as it requires checking each element of the domain to see if it satisfies the query or not. We introduce a class of queries, which we call \emph{predicate queries}, also presented in~\citep{mckenna}, for which we can efficiently determine if the domains overlap. We then show how Theorem~\ref{thm:arbitrarycomposition} relates to this query class.


\subsection{Predicates and Predicate Queries}
\label{sub:queries}

A \emph{predicate} on an attribute $A$ is a boolean function $\phi: A \rightarrow \{0, 1\}$. An attribute value $a \in A$ is said to \emph{satisfy} a predicate $\phi$ if $\phi(a) = 1$. The \emph{coverage} $C_\phi(A)$ of a predicate $\phi$ on the attribute $A$ is the set of all attribute values of $A$ that satisfy $\phi$, i.e.,
\[
C_\phi(A) \coloneqq \{a \in A : \phi(a) = 1\}.
\]
Two predicates $\phi_1$ and $\phi_2$ are \emph{disjoint} on attribute $A$ if $C_{\phi_1}(A) \cap C_{\phi_2}(A) = \emptyset$. Otherwise they are said to \emph{overlap}. 
\iffull
A \emph{tautology} on the attribute $A$, denoted $I$, is the predicate whose coverage on $A$ is $A$ itself. A \emph{contradiction} on the attribute $A$, denoted $I^c$, is the predicate whose coverage on $A$ is empty. We say that a predicate is \emph{non-trivial} if it is neither a tautology nor a contradiction. The following proposition is straightforward.
\begin{proposition}
\label{prop:trivial}
Let $I$ be a tautology, $I^c$ be a contradiction and $\phi$ be an arbitrary predicate on $A$. Let $A$ be an attribute. Then
\begin{enumerate}
    \item $I$ and $\phi$ overlap on $A$, except if $\phi$ is a contradiction.
    \item $I^c$ and $\phi$ are disjoint on $A$. 
\end{enumerate}
\end{proposition}
\begin{proof}
See Appendix~\ref{app:proofs}.
\end{proof}

\begin{example}
Consider the binary attribute $A = \{\texttt{Child}, \texttt{Adult}\}$, where $\texttt{Child}$ is anyone with age less than or equal to $18$, and an $\texttt{Adult}$, otherwise. Following are examples of predicates:
\begin{itemize}
    \item $\phi_1: A\texttt{ == Child}$ is a non-trivial predicate.
    \item $\phi_2: A\texttt{ == Child or Adult}$ is a tautology.
    \item $\phi_3: \texttt{any value of}\:A$ is a tautology.
    \item $\phi_4: \texttt{neither Child nor Adult}$ is a contradiction.
    \item $\phi_5: A\texttt{ == Dinosaur}$ is a contradiction.
\end{itemize}
\qed
\end{example}

It may seem rather pedantic to consider contradictions, as they necessarily include predicates with conditions outside the data domain (e.g., predicate $\phi_4$ in the example above). But including them in the framework frees us from having to restrict the query interface. As long as a query is defined as a conjunction of predicates on individual attributes (see the following), it is a valid query in our framework.
\fi

\descr{Predicate queries.} Following~\citep{mckenna}, we define a \emph{predicate query} $q$ on a row as a conjunction of $m$ predicates where the $i$\textsuperscript{th} predicate is evaluated on the $i$\textsuperscript{th} attribute value of the row. That is, given $x \in \mathbb{D}$,
\begin{equation}
\label{eq:query}
q(x) \coloneqq \phi_1(x_1) \wedge \phi_2(x_2) \wedge \cdots \wedge \phi_m(x_m)\text{.}
\end{equation}

Overloading notation, the query $q$ on a data set $D$ is defined as $q(D) \coloneqq \sum_{x \in D} q(x)$. One may write a query in terms of its constituent predicates: $q \coloneqq (\phi_1, \phi_2, \ldots, \phi_m)$. 
\iffull
More often than not, one may only be interested in a few attributes of the data set, and hence predicates may only be defined for those specific attributes. In this case, one can represent it in the above fashion by introducing tautologies for the remaining attributes in a straightforward manner.
\fi

\descr{Query coverage.} Since a conjunction of predicates is itself a predicate, one can view a query $q$ on a domain $\mathbb{D}$ as a predicate. With this, one can extend the notion of coverage to the domain. A row $x \in \mathbb{D}$ is said to \emph{satisfy} a query $q$ if $q(x) = 1$. It follows that a row $x \in \mathbb{D}$ satisfies a query $q \coloneqq (\phi_1, \phi_2, \ldots, \phi_m)$ if and only if for each $i$, $x_i$ satisfies $\phi_i$. The coverage $C_q(\mathbb{D})$ of a query $q$ on the domain $\mathbb{D}$ is the set of all rows that satisfy $q$, i.e., 

\[
C_q(\mathbb{D}) \coloneqq \{x \in \mathbb{D} : q(x) = 1\}.
\]

Our definition of predicate queries is broader than the set of statistical range queries considered in~\citep{xiao-np, inan-sensitivity}, and includes them as a special case, i.e., when each non-trivial predicate on an attribute is a range of values of the attribute. 
Two queries $q_1$ and $q_2$ are said to be \emph{disjoint} if $C_{q_1}(\mathbb{D}) \cap C_{q_2}(\mathbb{D}) = \emptyset$. Otherwise they are said to \emph{overlap}. 

\begin{proposition}
\label{prop:disjoint}
Two queries $q_1 \coloneqq (\phi_{1,1}, \phi_{2, 1}, \ldots, \phi_{m, 1})$ and $q_2 \coloneqq (\phi_{1,2}, \phi_{2, 2}, \ldots, \phi_{m, 2})$ are disjoint if and only if there exists at least one attribute $A_i$ such that $\phi_{i, 1}$ and $\phi_{i, 2}$ are disjoint on $A_i$.
\end{proposition}
\begin{proof}
See Appendix~\ref{app:proofs}. 
\end{proof}

Due to the set-theoretic nature of the notion of coverage, the results extend to any finite set of queries $Q$. In particular, we define $C_{Q}(\mathbb{D}) \coloneqq \bigcap_{q \in Q} C_q(\mathbb{D})$ to be the joint coverage of all queries in $Q$.

\descr{Computational Efficiency.}
\iffull
Propositions~\ref{prop:trivial} and~\ref{prop:disjoint} give
\fi
one an efficient way to decide whether two queries $q_1$ and $q_2$ are disjoint: for each attribute where the corresponding predicates of both queries are non-trivial, one checks if the two predicates are disjoint; if they are, the queries are disjoint, otherwise they overlap.
Assuming that the evaluation of a predicate on an attribute is efficient, the above procedure takes $\mathcal{O}(m)$ time only, as compared to the naive way of evaluating queries on each element of the domain, which takes time $\mathcal{O}(|\mathbb{D}|)$, which is exponential in $m$.

\descr{Generalized Query Coverage.} The above defined notion of query coverage is specific to the class of predicate queries, which is the main focus of this paper. However, one can define the notion more generally for other types of queries. Let $Q$ be a set of queries, where each $q \in Q$ is an arbitrary function $q: \mathbb{N}^{|\mathbb{D}|} \rightarrow \mathbb{R}$. Given a data set $D$ and a row $x \in D$, let $D_{\neg{x}}$ denote the neighboring data set of $D$ with one instance of $x$ removed from $D$. Given a row $x \in \mathbb{D}$, we say that $q$ \emph{covers} $x$ if there exists at least one data set $D$ such that $x \in D$ and $q(D) \neq q(D_{\neg{x}})$. The coverage, $C_q(\mathbb{D})$, of $q$ is defined as
\[
C_q(\mathbb{D}) \coloneqq \{x \in \mathbb{D} : q \text{ covers } x\}.
\]
The coverage of $Q$ is then defined as intersection of the coverage of all its queries, as before. Note that in the above, the amount of change in the answers is not specified. For an example where the generalized notion of query coverage deviates from query coverage for predicate (or count) queries, consider \emph{sum} queries, i.e., queries that sum the values of an attribute satisfying a given criterion (e.g., the salaries of all female managers in a company). Clearly, the absolute difference of the answers to any given sum query on two neighboring data sets (i.e., data sets that differ only in the inclusion/exclusion of a single row) depends on the row being removed. This is not the case with the predicate or counting queries; for such queries, if there is a change in answer, then the absolute difference is always $1$.

\subsection{Maximum Overlap}
\label{sub:max-overlap}
Let $Q \coloneqq \{q_1, q_2, \ldots, q_t\}$ be a set of $t$ queries. The \emph{maximum overlap} of $Q$, denoted $\gamma(Q)$, is defined by

\[
\gamma(Q) \coloneqq \max_{Q' \subseteq Q} \{|Q'| : C_{Q'}(\mathbb{D}) \neq \emptyset\}.
\]

It is easy to see that $1 \le \gamma(Q) \le t$. 

\begin{example}
Consider the set of queries $Q \coloneqq \{q_1, q_2, q_3\}$ defined by
\begin{align*}
    q_1 &: \texttt{Postcode == A, Native == Y} \\
    q_2 &: \texttt{Postcode == A OR B} \\
    q_3 &: \texttt{Postcode == B, Native == N}\text{.}
\end{align*}
Then
\begin{itemize}
    \item $q_1$ and $q_2$ are overlapping because the predicate \texttt{Native == Y} overlaps with the tautology \texttt{Native == Any} (not explicit) in $q_2$, and the predicates \texttt{Postcode == A} and \texttt{Postcode == A OR B} also overlap;
    \item $q_2$ and $q_3$ are overlapping because the predicate \texttt{Native == N} overlaps with the tautology \texttt{Native == Any} (not explicit) in $q_2$, and the predicates \texttt{Postcode == A OR B} and \texttt{Postcode == B} also overlap; and
    \item $q_1$ and $q_3$ are disjoint because the predicates \texttt{Native == Y} and \texttt{Native == N} are disjoint. 
\end{itemize}
We immediately have that $\gamma(Q) = 2$.
\qed
\end{example}

As mentioned earlier, we seek an efficient way to determine $\gamma(Q)$ as a function of the number of queries $t$ in $Q$. The naive way is to go through all subsets of $Q$ to determine $\gamma(Q)$, which takes time $\mathcal{O}(2^t)$. The following proposition sheds light on the difficulty of the problem.

\begin{proposition}
\label{prop:pairwise}
Let $Q$ be a set of queries. Then
\begin{enumerate}
    \item if $C_Q(\mathbb{D}) \neq \emptyset$, then all queries in $Q$ pairwise overlap; and
    \item it is possible that all queries in $Q$ pairwise overlap while $C_Q(\mathbb{D}) = \emptyset$.
\end{enumerate}
\end{proposition}
\begin{proof}
See Appendix~\ref{app:proofs}.
\end{proof}
Thus, we cannot determine $\gamma(Q)$ by simply checking pairs of queries to see if they overlap.

\descr{Maximum weight overlap.} We also consider the general case where each query $q \in Q$ has an associated weight $w: Q \rightarrow \mathbb{R}^+$. The weight corresponds to the privacy budget allocated to the query when a differentially private mechanism is used to answer $q$. To handle the case when multiple queries are answered by the mechanism, we need to define how their weights compose. The exact form of composition depends on the type of differential privacy used, e.g., $\epsilon$-DP or $f$-DP. However, there are common properties.
\begin{definition}[Composition Function]
\label{def:comp-func}
Let $Q$ be a set of queries, where each query $q \in Q$ has weight $w(q)$, for some function $w: Q \rightarrow \mathbb{R}^+$. A \emph{composition function} is a function $\mathrm{comp}: \mathcal{P}(Q) \rightarrow \mathbb{R}^+$ such that
\begin{itemize}
    \item $\mathrm{comp}(\{q\}) \geq w(q)$ for any $q \in Q$; and
    \item $\mathrm{comp}(Q') \geq \mathrm{comp}(Q'')$ if $Q' \supseteq Q''$, for any $Q', Q'' \subseteq Q$ (\emph{Monotonicity}).
\end{itemize}
\end{definition}

Examples of the function $\mathrm{comp}$ include a simple sum of weights (sequential composition under $\epsilon$-DP), a sum of squares of weights, or the square root of a sum of squares of weights (sequential composition under Gaussian DP). We define the \emph{maximum weight overlap} of $Q$, denoted $\gamma_w(Q)$, by

\[
\gamma_w(Q) \coloneqq \max_{Q' \in \mathcal{P}(Q)} \{\mathrm{comp}(Q') : C_{Q'}(\mathbb{D}) \neq \emptyset\}.
\]

Define the set $O_1$ as the \emph{set of overlapping queries} in $\mathcal{P}(Q)$, i.e., $O_1 \coloneqq \{Q' \in \mathcal{P}(Q) : C_{Q'}(\mathbb{D}) \neq \emptyset\}$. Also, define the set $O_2$, the \emph{set of maximal overlapping subsets of $O_1$}, by $O_2 \coloneqq \{Q' \in O_1 : \text{for all } Q'' \in O_1, Q' \not\subset Q''\}$, i.e., the set of elements of $O_1$ which are not proper subsets of any other elements in $O_1$. We have the following proposition.
\begin{proposition}
\label{prop:max-overlap}
$\gamma_w(Q) = \max\,\{\mathrm{comp}(Q') : Q' \in O_1\} = \max\,\{\mathrm{comp}(Q') : Q' \in O_2\}$. 
\end{proposition}
\begin{proof}
\iffull
The first equality follows immediately from the definitions of maximum weight overlap and the set $O_1$. We consider the second equality. Let $A_1 \coloneqq \{\mathrm{comp}(Q') : Q' \in O_1\}$, and let $A_2 \coloneqq \{\mathrm{comp}(Q') : Q' \in O_2\}$. Since $A_2 \subseteq A_1$, we have that $\max A_1 \geq \max A_2$. Next consider $\max A_2$. Let $Q' \in O_1$, and let $Q'' \supseteq Q'$, which is guaranteed to be in $O_2$ by construction. Then, by the monotonicity property of the composition function, we have that $\max A_2 \geq \mathrm{comp}(Q'') \geq \mathrm{comp}(Q')$. Thus, $\max A_2 \geq \max A_1$. From this it follows that $\max A_2 = \max A_1 = \gamma_w(Q)$.
\fi
\end{proof}

Note that no two distinct subsets of $O_2$ overlap, since otherwise their union will be in $O_2$, contradicting the fact that they are maximally overlapping subsets of queries. The proof of the above theorem is similar to the proof of Theorem~\ref{the:par-comp:gdp}. The advantage here is that one can directly compute maximum (weighted) overlap by considering overlapping queries and then use the underlying composition function, as long as the composition function allows parallel composition and the query weights are equal to the privacy parameter associated with each query. This decouples the computational problem from the underlying type of differential privacy.
For instance, if one considers $\mu_i$-GDP mechanisms, then $\mathrm{comp}(Q) = \sqrt{\sum_{q_i} w(q_i)^2}$, where $w(q_i) = \mu_i$. If one considers $\epsilon_i$-DP mechanisms under sequential composition, then $\mathrm{comp}(Q) = \sum_{q_i} w(q_i)$, where $w(q_i) = \epsilon_i$. And for homogeneous mechanisms, $\mathrm{comp}(Q) = \epsilon \cdot |Q|$, under basic composition of standard differential privacy. This last result follows from the following proposition, which can easily be proved by invoking the monotonicity property of the composition function.
\begin{proposition}
\label{prop:max-w-overlap}
Let $Q$ be a set of queries. If all queries in $Q$ have the same weight, then $\gamma_w(Q) = w \cdot \gamma(Q)$.
\end{proposition}

\subsection{Utility Gain}
\label{sec:util-gain}
Assume the data custodian wishes to release answers to a set $Q$ of $t$ queries via a differentially private mechanism $\mathcal{M}$. Let $Y_i$ denote the random variable representing the noise added to the $i$\textsuperscript{th} query by the differentially private mechanism, i.e., $Y_i = \mathcal{M}(q_i, D) - q_i(D)$. We are interested in the expectation of the absolute value of the total noise added over all $t$ queries. Under sequential composition, this is
\iffull
\[
\mathbb{E}\left(\sum_{i=1}^t |Y_i| \right) = \sum_{i = 1}^{t} \mathbb{E}(|Y_i|).
\]
\fi
Let $Q' \subseteq Q$ be the set such that $\gamma_w(Q) = \mathrm{comp}(Q')$. Under optimal composition, the expectation is
\iffull
\[
\mathbb{E}\left(\sum_{i : q_i \in Q'} |Y_i| \right) = \sum_{i : q_i \in Q'} \mathbb{E}(|Y_i|).
\]
\fi
The \emph{utility gain}, denoted $U$, is defined as:
\begin{equation}
\label{eq:util}
U \coloneqq 1 - \frac{\sum_{i : q_i \in Q'} \mathbb{E}(|Y_i|)}{\sum_{i = 1}^{t} \mathbb{E}(|Y_i|)}
\end{equation}
Thus, e.g., if $\mathcal{M}$ is the Laplace mechanism under basic composition of pure differential privacy, with $w(q_i) = \epsilon$ for all $i$, then Eq.~\ref{eq:util} simplifies to
\begin{equation}
\label{eq:util-simple}
    U = 1 - \frac{\gamma}{t}.
\end{equation}
Similarly, if $\mathcal{M}$ is the Gaussian mechanism with composition under $\mu$-GDP, with $w(q_i) = \mu$ for all $i$, then the utility gain is the same as above. Thus, the above metric is not dependent on the composition function, but only on the optimal use of parallel composition, and compares it directly to sequential composition. 

For example, if a set of $t=100$ queries has $\gamma = 70$, the utility gain is 30\%. For the data custodian, this means 30\% less noise needs to be added to query results while maintaining the same overall privacy budget. In some cases we shall also report the more commonly used average $l_1$-error, for ease of comparison against our utility gain metric. For a set $Q$ of $t$ queries, $q_1, q_2, \ldots, q_t$, answered via a differentially private mechanism $\mathcal{M}$, the average $l_1$-error is defined as follows:
\begin{equation}
\label{eq:l1}
\text{Average $l_1$ Error} \coloneqq \frac{1}{t} \sum_{i = 1}^t |\mathcal{M}(q_i, D) - q_i(D)|
\end{equation}

\section{Hardness of Maximum Overlap}
\label{sec:comp-theo}
Even if one can efficiently check whether two predicate queries overlap, finding the maximum overlap remains a hard problem. This is mainly because one needs to search the powerset of the set of queries $Q$. Indeed, in this section, we show that finding the maximum overlap of a set of predicate queries is
\iffull
NP-complete. 
\fi

This will be established by linking one instance of maximum weighted overlap with the problem of finding $l_1$-sensitivity of a set of queries $Q$. As mentioned in the introduction,~\cite{xiao-np} have already shown that computing $l_1$-sensitivity is NP-hard. This then implies readily that maximum weight overlap is NP-hard. 
\iffull
We also show an alternate proof that also establishes NP-completeness of the maximum overlap problem. 
\fi

Let $Q$ be a set of queries, and let $q \in Q$. Let $D$ and $D'$ be neighboring databases. The \emph{sensitivity} of the query $q$ is defined as $\Delta q \coloneqq \max_{D \sim D'} \lvert q(D) - q(D') \rvert$. The \emph{$l_1$-sensitivity} of $Q$ is defined as
\iffull
\[
\Delta Q \coloneqq \max_{D \sim D'} \left(\sum_{q \in Q} \lvert q(D) - q(D') \rvert \right).
\]
\fi
Next define $w(q) \coloneqq \Delta q$ for each $q \in Q$, and define the composition function as $\mathrm{comp}(Q) \coloneqq \sum_{q \in Q} w(q) = \sum_{q \in Q} \Delta q$. We next show that $\gamma_w(Q) = \Delta Q$ under a \emph{consistency condition}. More specifically, note that the notion of generalized query coverage (discussed in Section~\ref{sub:queries}) defines a query $q$ to cover some row $x$ in the domain if its answer on at least one data set $D$ containing $x$ differs from its answer on the neighboring data set $D_{\neg x}$. But this does not say how much the answer changes by, or whether the change is the same for all rows. We say that the set of queries $Q$ \textit{satisfies} the consistency condition if for each $q \in Q$ we have $\lvert q(D) - q(D') \rvert = \Delta q$ whenever $q(D) \neq q(D')$ for all neighboring data sets $D$ and $D'$. In other words, whenever there is a change in query value over two neighboring data sets, it is the same change over any two neighboring data sets, i.e., the maximum possible change. Thus, the equivalence of the two notions may not hold for general queries, i.e., without the consistency condition being satisfied.
\begin{theorem}
\label{the:mo-equals-sen}
For each $q \in Q$, if $\lvert q(D) - q(D') \rvert = \Delta q$ whenever $q(D) \neq q(D')$ for all neighboring data sets $D$ and $D'$, then $\gamma_w(Q) = \Delta Q$. 
\end{theorem}
\begin{proof}
\iffull
Let $Q'$ be the subset of $Q$ such that $\gamma_w(Q) = \sum_{q \in Q'} \Delta q$. Let $D$ and $D'$ be the neighboring data sets such that $\Delta Q = \sum_{q \in Q} |q(D) - q(D')|$. Let us assume that the row they differ in is $x$. Let $Q'' \subseteq Q$ be such that for all $q \in Q''$, $q(D) \neq q(D')$. Then, through the consistency condition
\[
\Delta Q = \sum_{q \in Q} |q(D) - q(D')| = \sum_{q \in Q''} |q(D) - q(D')| = \sum_{q \in Q''} \Delta q.
\]
Since all queries in $Q''$ cover $x$, $C_{Q''}(\mathbb{D}) \neq \emptyset$. Therefore, according to the definition of maximum overlap 
\[
\Delta Q = \sum_{q \in Q''} \Delta q \leq  \sum_{q \in Q'} \Delta q = \gamma_w(Q).
\]
Next take $Q'$, and let $x \in C_{Q'}(\mathbb{D})$. Let $D_x$ be a data set containing $x$, and $D_{\neg x}$ be the neighboring data set of $D_x$ with one instance of $x$ removed. Once again, according to the consistency condition and the definition of maximum overlap, we have
\begin{align*}
   \gamma_w(Q) &= \sum_{q \in Q'} \Delta q  \\
               &= \sum_{q \in Q'} \lvert q(D_x) - q(D_{\neg x}) \rvert \\
               &= \sum_{q \in Q} \lvert q(D_x) - q(D_{\neg x}) \rvert \\
               &\leq \max_{D \sim D'} \sum_{q \in Q} \lvert q(D) - q(D') \rvert \\
               &= \Delta Q.
\end{align*}
Hence $\gamma_w(Q) = \Delta Q$.
\fi
\end{proof}
In particular, the predicate queries considered in this paper, and the statistical range queries~\citep{xiao-np, inan-sensitivity} (a proper subset of the former) are examples of queries that obey the consistency condition. Note that Theorem~\ref{the:mo-equals-sen} only shows the equivalence of the two notions under basic composition: weights add up linearly. The notion of maximum weight overlap is more general than $l_1$-sensitivity and encompasses other forms of composition, e.g., composition of Gaussian mechanisms under $f$-DP. The Gaussian mechanism is not $\epsilon$-DP under $l_1$-sensitivity (it is $(\epsilon, \delta)$-DP under $l_2$-sensitivity). Hence, our notion and accompanying results have broader applicability.
We define the maximum overlap problem in terms of predicate queries.

\medskip
\textsc{Maximum Overlap}: Given a set $\mathcal{A}$ of $m$ attributes, a set $\Phi$ containing a predicate $\phi_A$ for each attribute $A\in \mathcal{A}$, a set $Q$ of $t$ predicate queries $q_1, q_2, \dots, q_t$, and a positive integer $k < t$, is there a subset of $k$ or more queries that overlap?
\medskip


The following then follows immediately from Theorem~\ref{the:mo-equals-sen} and the NP-hardness of $l_1$-sensitivity of the statistical range queries~\citep{xiao-np}.
\begin{theorem}[\citep{xiao-np}]
\label{theo:overlap-np}
\textsc{Maximum Overlap} is NP-hard.
\end{theorem}

It follows that the \textsc{Maximum Weight Overlap} problem is NP-hard as well; for otherwise it could be used to efficiently solve the \textsc{Maximum Overlap} problem with the same weight assigned to all queries (Proposition~\ref{prop:max-w-overlap}). 

\iffull


\descr{NP-completeness and alternative proof of NP-hardness.}
The above proof shows only that the problem is NP-hard. We can in fact show that the decision problem \textsc{Maximum Overlap} is NP-complete. The proof from~\cite{xiao-np} uses the \textsc{Max 2Sat} problem.
Our alternative proof establishes a connection with a graph problem, and also proves membership in NP. Our reduction is from the NP-complete \citep{GareyJ79} \textsc{Max Cut} problem.

\medskip
\textsc{Max Cut}: Given a graph $G=(V,E)$ and an integer $k$, is there a vertex subset $S\subseteq V$ such that the number of edges with one endpoint in $S$ and the other endpoint in $V\setminus S$ is at least $k$?
\medskip

\begin{theorem}
	\textsc{Maximum Overlap} is NP-complete.
\end{theorem}
\begin{proof}
	First, we argue that \textsc{Maximum Overlap} is in NP.
	For an instance $(\mathcal{A}, \Phi, Q, k)$ of the \textsc{Maximum Overlap} problem with $|Q|=t$, a potential certificate $Q'$ is a subset of at least $k$ and at most $t$ queries from $Q$. The size of $Q'$ is polynomial in the input size since it is a subset of the input.
	A potential certificate $Q'$ is a valid certificate if $C_{Q'}(\mathbb{D})\ne \emptyset$. Validity of a potential certificate can be checked in polynomial time by computing the conjunction of the queries in $Q'$, which can be done by considering the conjunction of predicates on each attribute individually.
	
	NP-hardness is established by reduction from \textsc{Max Cut}. Given an instance $(G=(V,E),k)$ of \textsc{Max Cut}, we compute an instance $(\mathcal{A},\Phi,Q,k')$ of \textsc{Maximum Overlap} as follows. For each vertex $v\in V$, we have an attribute $x_v=\{0,1\}$ and a predicate $\phi_v$ acting on the attribute $x_v$. We think of $\phi_v(x_v)=1$ if $v\in S$ in the \textsc{Max Cut} instance and $\phi_v(x_v)=0$ otherwise.
	For each edge $uv\in E$, we add the queries $\overline{\phi_u(x_u)} \wedge \phi_v(x_v)$ and $\phi_u(x_u) \wedge \overline{\phi_v(x_v)}$.
	Set $k' = k$.
	We will show that $G$ has a vertex subset $S$ such that $|\{uv\in E: u\in S, v\in V\setminus S\}|\ge k$ if and only if there is a subset of at least $k'$ queries that overlap.
	
	First, consider any $S\subseteq V$ such that $|\{uv\in E: u\in S, v\in V\setminus S\}|\ge k$. We use predicates where $\phi_v(x_v)=1$ if and only if $v\in S$.
	Consider each edge $uv\in E$. If $|S\cap \{u,v\}|=1$, then we have a contribution of 1 to $|\{uv\in E: u\in S, v\in V\setminus S\}|$ and exactly one query among $\overline{\phi_u(x_u)} \wedge \phi_v(x_v)$ and $\phi_u(x_u) \wedge \overline{\phi_v(x_v)}$ is satisfied. Otherwise, if $|S\cap \{u,v\}|\ne 1$, then we have a contribution of 0 to $|\{uv\in E: u\in S, v\in V\setminus S\}|$ and neither query among $\overline{\phi_u(x_u)} \wedge \phi_v(x_v)$ and $\phi_u(x_u) \wedge \overline{\phi_v(x_v)}$ is satisfied.
	Therefore, at least $k'=k$ queries overlap if and only if $|S|\ge k$.
	
	Second, assume that at least $k'$ queries overlap. By the same reasoning, we conclude that there is a subset $S\subseteq V$ with $|\{uv\in E: u\in S, v\in V\setminus S\}|\ge k$.
	
	Since the construction can be done in polynomial time and produces an equivalent instance of \textsc{Maximum Overlap}, we conclude that \textsc{Maximum Overlap} is NP-hard.
\end{proof}

\fi

\section{Connection to Graphs}
\label{sec:graphtheory}

Inan et al.~\cite{inan-sensitivity} relate the problem of computing $l_1$-sensitivity of a set of queries $Q$ to a graph problem, by modelling $Q$ as a graph. The authors then upper bound computing $l_1$-sensitivity to finding the maximum clique of the graph. The advantage is that we can use well-known graph algorithms to solve the problem. Likewise, in this section, we shall represent the maximum overlap problems in terms of 
\iffull
graphs. Even though maximum overlap (and computing $l_1$-sensitivity) is an NP-hard problem, there are efficient algorithms (in practice) to solve the related graph problems. 
\fi
There are two key differences between~\citep{inan-sensitivity} and our treatment in this section:
\begin{itemize}
    \item We show that the problem of finding maximum weighted overlap exactly translates to a hypergraph problem, but with a computationally expensive solution. Details appear in Appendix~\ref{app:hypergraph}.
    \item Due to the expensive nature of solving the hypergraph problem, we instead target pairwise overlaps of queries using simple graphs, which allows us to upper bound maximum overlap with the clique number of the graph, in the manner of ~\citep{inan-sensitivity}. However, unlike~\citep{inan-sensitivity}, we show that maximum overlap is further upper bounded by the chromatic number of the graph (to be defined shortly). The advantage here is that whereas here and in~\citep{inan-sensitivity}, one is forced to use exact algorithms to compute the clique number, lest the maximum overlap be underestimated (leading to a potential privacy risk), the chromatic number can be computed using an approximate algorithm that never underestimates the maximum overlap. This allows us to use these approximate algorithms for query set sizes and domain sizes significantly beyond what can feasibly be handled by the exact clique algorithms.
\end{itemize}

A graph is a pair $G \coloneqq (V, E)$, where $V$ is a set of vertices and $E$ is a set of edges such that $E \subseteq V \times V$. We consider the graph formed by pairwise overlaps of queries, rather than all possible subsets of queries. This pairwise query graph takes time $\mathcal{O}(t^2)$ to construct.

\descr{Query graph.} Given a set of queries $Q \coloneqq \{q_1, q_2, \ldots, q_t\}$, their query graph, $\mathcal{G}(Q) \coloneqq (V, E)$, is defined as follows: each query is a vertex (i.e., $V = Q$), and two vertices have an edge connecting them if the queries they represent overlap: 
\[
E \coloneqq \{ (q_i, q_j) : q_i \text{ and } q_j \text{ overlap}\}.
\]
See Appendix~\ref{app:hypergraph} for an example set of queries and its query graph.

\descr{Weighted query graph.} A weighted query graph is a query graph where each vertex is assigned a weight $w : Q \rightarrow \mathbb{R}^+$, and weights compose via a composition function $\mathrm{comp}: \mathcal{P}(Q) \rightarrow \mathbb{R}^+$ function (see Definition~\ref{def:comp-func}). We denote this by $\mathcal{G}_w(Q)$.

Whilst the query graph is far faster to construct than the overlap hypergraph (Appendix~\ref{app:hypergraph}), there is no analogue of Proposition \ref{prop:rankvsgamma} (in the appendix), which exactly links maximum overlap to cardinality of the largest hyperedge of the hypergraph. However, we present two graph metrics that bound the maximum overlap from above --- the clique number and chromatic number.

\descr{Clique number.} A \emph{complete} subgraph $G'$ is a subgraph of $G$ where all vertices are pairwise adjacent~\citep{diestel-gt}. The \emph{clique number} of a graph $G$ is defined as the size of the largest complete subgraph of $G$.

\descr{Weighted clique number.} The \emph{weighted} clique number of the weighted query graph $\mathcal{G}_w(Q)$ is defined as the maximum value of $\mathrm{comp}(Q')$, where $Q' \subseteq Q$ is a complete subgraph.

To avoid notational clutter, we will denote the clique number $\omega$ of the query graph $\mathcal{G}(Q)$ and the weighted query graph $\mathcal{G}_w(Q)$ by $\omega(Q)$ and $\omega_w(Q)$, respectively. 

\begin{proposition}
\label{prop:cliquevscover}
Let $Q$ be a set of queries. Then $\gamma_w(Q) \leq \omega_w(Q)$.
\end{proposition}
\begin{proof}
\iffull
Recall the definition of $\gamma_w(Q)$:
\[
\gamma_w(Q) \coloneqq \max_{Q' \subseteq Q} \{\mathrm{comp}(Q') : C_{Q'}(\mathbb{D}) \neq \emptyset\}.
\]
By part (1) of Proposition~\ref{prop:pairwise}, all queries in $Q'$ pairwise overlap, and hence form a complete subgraph of the query graph. By part (2) of Proposition~\ref{prop:pairwise}, it is also possible for a clique $Q''$ on the query graph to contain queries that all pairwise overlap, but have $C_{Q''}(\mathbb{D}) = \emptyset$. 
By the monotonicity property of $\mathrm{comp}$, we must have that $\omega_w(Q)$ is bounded from below by $\gamma_w(Q)$.
\fi
\end{proof}

We also have the same bound in the unweighted case:

\begin{corollary}
Let $Q$ be a set of queries. Then $\gamma(Q) \leq \omega(Q)$.
\end{corollary}

Next, we introduce the chromatic number of $\mathcal{G}(Q)$.

\descr{Independent set.} For a graph $\mathcal{G} \coloneqq (V, E)$, an \emph{independent set} is a subset of vertices $S \subseteq V$ such that no two vertices $v_i, v_j \in S$ share an edge \citep{voloshin2009introduction}. Let $S$ be an independent set of queries in the query graph $\mathcal{G}_w(Q)$. We define the \emph{weight} of $S$ as $w(S) \coloneqq \max\{w(q) : q \in S\}$. This is consistent with the fact that these queries do not overlap, and hence can be composed in parallel if given as input to a differentially private mechanism that allows parallel composition.


\descr{Chromatic number and minimum weight coloring.} A \emph{proper coloring} of a graph $\mathcal{G} \coloneqq (V, E)$ is a partition $\mathcal{S} \coloneqq (S_1, S_2, \ldots, S_k)$ of $V$ into $k$ independent sets. The \emph{chromatic number} $\chi(G)$ of the graph is defined as the minimum $k$ over all proper colorings \citep{voloshin2009introduction}. Another (more common) definition of the chromatic number is the minimum number of colors needed to color the vertices, such that no two adjacent vertices have the same color.
As defined above,
each independent set has weight $w(S_i)= \max \{w(q) : q \in S_i\}$. The weight of a coloring $\mathcal{S}$ is then given by the sequential composition of the weights of the independent sets:
\[
w(\mathcal{S}) = \mathrm{comp}(\mathcal{S}) = \sum_{i = 1}^k w(S_i) = \sum_{i = 1}^k \max \{w(q) : q \in S_i\}.
\]
A minimum weight coloring $\chi_w(G)$ is then the coloring $\mathcal{S}$ of $\mathcal{G}_w(Q)$ that minimizes $w(\mathcal{S})$. This is called the chromatic number, which we shall denote in the unweighted case by $\chi(Q)$, and in the weighted case by $\chi_w(Q)$.

It is well known that $\omega(Q) \leq \chi(Q)$, and in fact this gap can be arbitrarily large \citep{mycielski1955coloriage}. Similarly, $\omega_w(G) \leq \chi_w(G)$. From Proposition \ref{prop:cliquevscover}, it follows that:
\begin{equation}
\label{eq:approx}
\gamma_w(Q) \leq \omega_w(Q) \leq \chi_w(Q) \leq \mathrm{comp}(Q).    
\end{equation}
And in the unweighted case, we have 
\iffull
\[
\gamma(Q) \leq \omega(Q) \leq \chi(Q) \leq |Q|.
\]
\fi
Thus, computing the (weighted) clique number or chromatic number of the query graph will give an approximation for the (weighted) maximum overlap. We give algorithms for computing these metrics in the next section.

\section{Computing Maximum Overlap}
\label{sec:algorithms}
In this section we present a number of algorithms for computing the maximum overlap, clique number and chromatic number of a set of queries. By Theorem \ref{thm:arbitrarycomposition}, maximum overlap is exactly the privacy loss of the set of queries, and according to Eq.~\ref{eq:approx}, the clique and chromatic number are an approximation (overestimate) of the privacy loss.

\descr{`Safe' approximations for maximum overlap.} The problems of computing $\omega(Q)$ and $\chi(Q)$ are known to be NP-hard \citep{lewis2015guide, voloshin2009introduction}. In Theorem~\ref{theo:overlap-np}, we show that computing $\gamma(Q)$ is also NP-hard. Thus, as the query set grows, computing $\omega(Q), \chi(Q)$ or $\gamma(Q)$ will become infeasible. We therefore consider approximate algorithms.

Since the clique number $\omega(Q)$ is framed as a maximization problem (namely, the problem of finding the largest clique), any approximate $\tilde \omega(Q)$ will be upper bounded by $\omega(Q)$. However, according to Eq.~\ref{eq:approx}, this means it may be possible that $\tilde \omega(Q) \leq \gamma(Q)$, which may lead to a privacy leakage! As such, we say that it is `unsafe' to use an approximate clique number, i.e., $\tilde \omega(Q)$. 

By contrast, any approximate chromatic number $\tilde \chi(Q)$ will satisfy $\chi(Q) \leq \tilde \chi(Q) \leq |Q|$, since it is framed as a minimization problem. Therefore, $\gamma(Q) \leq \tilde \chi(Q)$. This makes it `safe' to compute an approximate chromatic number, as there is never a risk of privacy leakage. Therefore, we are bound to consider exact algorithms for maximum clique, whereas for chromatic number we can use more efficient approximate algorithms.

\iffull
\subsection{Maximum Clique}
\fi
Computing the maximum clique of an arbitrary graph is an extensively studied problem. A detailed, recent review is given in \citep{wu2015review}, which discusses both approximate and exact computation of $\omega_w(Q)$. Due to the necessity of ensuring safe approximation, we consider only algorithms for exactly computing $\omega_w(Q)$. These algorithms are based on the branch-and-bound framework, which consists of two main aspects --- a search strategy to recursively partition the search space into smaller sub-problems (branching), and a pruning strategy that allows sub-problems with a provably sub-optimal solution to be pruned from the search space (bounding) \citep{clausen1999branch, morrison2016branch}. 
\iffull
Many pruning strategies are very effective at reducing the size of the search space that needs to be explored. As a result, many branch-and-bound algorithms perform well in practice, despite a lack of theoretical results about their performance.

\begin{algorithm}[!ht]
\SetAlgoLined
\SetAlCapSkip{1em}
\DontPrintSemicolon{}
\let\oldnl\nl
\newcommand{\nonl}{\renewcommand{\nl}{\let\nl\oldnl}}
Initialize, $G \leftarrow \mathcal{G}(Q)$, $X \leftarrow \emptyset$, $B \leftarrow \emptyset$, $\mathrm{ub} \leftarrow \mathrm{comp}(Q)$.

\texttt{MaxWeightClique} ($G$, $X$, $B$, $\mathrm{ub}$):\;
\If{$G = \emptyset$}{
    \Return $X$.\;
}
$\tilde \chi_w(G) \leftarrow \text{an approximate coloring of } G$.\;
$\mathrm{ub}' \leftarrow \min(\mathrm{ub}, \mathrm{comp}(X) + \tilde \chi_w(G))$.\;
\If{$\mathrm{ub}' \leq \mathrm{comp}(B)$}{
    \Return $B$.\;
}
$q \leftarrow \text{a max degree vertex of } G$.\;
$G' \leftarrow \text{graph induced by } N(q)$.\;
$X' \leftarrow X \cup \{q\}$.\;
$B' \leftarrow \texttt{MaxWeightClique}(G',X', B, \mathrm{ub}')$.\;
\If{$\mathrm{ub}' = \mathrm{comp}(B')$}{
    \Return $B$.\;
}
$G'' \leftarrow \text{graph induced by } V(G) - \{q\}$.\;
\Return $\texttt{MaxWeightClique}(G'', X, B, \mathrm{ub}')$.\;
\caption{Maximum Clique Algorithm with coloring-based pruning \citep{coudert1997exact}}
\label{alg:maxcliquekcolor}
\end{algorithm}


For our experiments, we implement the maximum clique algorithm presented in \citep{coudert1997exact}. A description of the algorithm is given in Algorithm \ref{alg:maxcliquekcolor}. The algorithm takes as input the query graph $\mathcal{G}(Q)$, a candidate maximum clique $X$, the current best known clique $B$ and an upper bound on the maximum weight of the clique $\mathrm{ub}$. $X$ and $B$ are initialized as empty sets, and $\mathrm{ub}$ is initialized as $\mathrm{comp}(Q)$. In the algorithm $N(q)$ denotes the neighbors of a query $q$ in the query graph. The algorithm recursively builds a maximum clique by selecting maximum degree nodes from the query graph and pruning nodes that are not adjacent to the currently selected nodes in $X$. This strategy quickly finds a candidate maximum clique, which is set to $B$ when no nodes remain in $\mathcal{G}(Q)$. This candidate maximum clique forms a lower bound on the weight of the true maximum clique.

The algorithm then backtracks, recursively exploring the search space to find a larger weighted clique than $B$. To prune the search space, an upper bound on the maximum weight of $X$ is computed by adding the current weight $\mathrm{comp}(X)$ and an approximate coloring of the remaining query graph $\tilde \chi_w(G)$. If this upper bound is smaller than $\mathrm{comp}(B)$, there is no point in searching further, allowing the algorithm to prune and backtrack. This algorithm could be further improved. For example, the authors of \citep{konc2007improved} note that there is a trade-off between the run-time cost of computing $\tilde \chi_w(G)$ and the level of pruning performed at different recursion depths. Several other optimized algorithms for maximum clique are discussed in \citep{wu2015review}.
\fi

\iffull
\subsection{Maximum Overlap}
\fi
Given the connection between maximum overlap and maximum clique, we can adapt Algorithm \ref{alg:maxcliquekcolor} to compute the (exact) maximum overlap directly. This can be done simply by adding the constraint that the maximum clique returned must have non-empty intersection prior to line 3 in Algorithm \ref{alg:maxcliquekcolor}, i.e., if $C_{X}(\mathbb{D}) = \emptyset$ then return $B$. 
This change ensures we compute the maximum weighted overlap rather than the maximum weighted clique. 
\iffull
It may be possible to develop further optimized algorithms for maximum overlap based on newer maximum clique algorithms, but we defer this task to future work.
\fi

\iffull
\subsection{Approximate Chromatic Number}
\fi
There is significant literature dedicated to the problem of computing an exact or approximate chromatic number (see, e.g., \citep{lewis2015guide}). Due to the NP-hardness of exactly computing the chromatic number and the `safe' approximation issue discussed earlier, we focus on algorithms for computing an approximate chromatic number $\tilde \chi_w(Q)$. One such algorithm is the DSatur algorithm \citep{brelaz1979new}. For our experiments, we use the implementation of DSatur available in the Python \texttt{networkx} package \citep{SciPyProceedings_11}. The DSatur algorithm is presented as 
\iffull
Algorithm \ref{alg:approxcoloring}.

\begin{algorithm}[!ht]
\SetAlgoLined
\SetAlCapSkip{1em}
\DontPrintSemicolon{}
\let\oldnl\nl
\newcommand{\nonl}{\renewcommand{\nl}{\let\nl\oldnl}}
\texttt{DSatur}($X \leftarrow Q$, $\mathcal{S} \leftarrow \emptyset$).\;
\While{$X \neq \emptyset$}{
    choose $q \in X$ with maximal saturation degree.\;
    \For{$j \leftarrow 1$ to $|\mathcal{S}|$}{
        \eIf{$S_j \cup \{q\}$ is an independent set}{
            $S_j \leftarrow S_j \cup \{q\}$.\;
            \textbf{break}\;
        }
        {
            $j \leftarrow j + 1$
        }
    }
    \If{$j > |\mathcal{S}|$}{
        $S_j \leftarrow \{q\}$.\;
        $\mathcal{S} \leftarrow \mathcal{S} \cup S_j$.\;
    }
    $X \leftarrow X - \{q\}$.\;
}
\Return $\mathcal{S}$
\caption{DSatur Algorithm for the approximate chromatic number of a graph \citep{lewis2015guide}}
\label{alg:approxcoloring}
\end{algorithm}


This algorithm takes as input the set of queries $Q$ (and their query graph $\mathcal{G}(Q)$), as well as an empty partition $\mathcal{S}$. The algorithm returns a valid coloring $\mathcal{S}$. Lines 4-14 of this algorithm comprise a greedy algorithm for finding a coloring of a graph. The algorithm simply chooses vertices one by one, and checks to see if the vertex can be added to any existing independent sets in $\mathcal{S}$ (lines 4-9). If it cannot, the vertex becomes a new independent set (lines 10-12). Once all vertices have been placed into $\mathcal{S}$, the algorithm terminates.

What separates DSatur from a typical greedy algorithm is the heuristic on line 3 for selecting vertices. The saturation degree of an uncolored vertex $v$ is defined as the number of different colors assigned to adjacent vertices. Thus, a vertex with maximal saturation degree can be considered as one that has the fewest available colors from which to choose. In practice, this heuristic works very well, and for certain classes of graphs the DSatur algorithm produces an optimal coloring. In the worst case, the DSatur algorithm has run time $\mathcal{O}(t^2)$, where $t$ is the number of vertices in the graph \citep{lewis2015guide}.
\fi

\section{Experimental Evaluation}
In this section, we demonstrate the effectiveness of our approach in terms of computational efficiency and utility gain as a function of the domain size and number of queries on both synthetic and real-world data sets. Our use case is the setting where an analysts asks queries through an online interface. We therefore fix a cap of 60 seconds on the amount of time it should take for the algorithm to return a solution to the problem, i.e., maximum overlap or approximate maximum overlap. We demonstrate that the approximate chromatic number algorithm can find an approximation to the maximum overlap within this time bound for a much larger set of queries than the exact clique number algorithm and its maximum overlap variant. At the same time, in all cases, we find that the gap between the approximate chromatic number and the maximum overlap is very small. Thus, this demonstrates the feasibility of computing the maximum overlap using the approximate chromatic number algorithm. 

\subsection{Effect of Domain Size and Number of Queries}
\label{subsec:scalingexperiments}

We first experiment with scaling the domain size and number of queries to assess the feasibility of the algorithms discussed in Section \ref{sec:algorithms}. More specifically, we use the exact clique number and its exact maximum overlap variant based on Algorithm~\ref{alg:maxcliquekcolor} and the approximate chromatic number algorithm based on the \texttt{DSatur} algorithm presented in Algorithm~\ref{alg:approxcoloring}. We consider a varying number of queries $t$ and a varying domain size $|\mathbb{D}|$. We consider the algorithm to have completed if it provides a result in under 60 seconds. Otherwise, the algorithm is considered to be too time expensive and recorded as a `time-out'. For each of the experiments, we select, uniformly at random, a number of attributes $m$ ranging from $10$ to $50,000$, and then select for each attribute $10^k$ attribute values, where $k$ is chosen uniformly at random from the set $\{1, 2, \ldots, 6\}$. The domain size, which is capped at $10^{80000}$ for reasons of computational feasibility, is then calculated as the product of the sizes of the sets of attribute values.
\iffull
Thus, the `log10 Domain Size' on the y-axis for each of the following plots can also be considered as the number of attributes $m$. 
\fi

\descr{Uniform distribution.} 
A single query on a given domain $\mathbb{D}$ is generated by selecting a random number of attributes $m' \sim \text{Uniform}(1,m)$. For each of the $m'$ attributes selected, we construct a predicate by randomly selecting a subset of values of size $a' \sim \text{Uniform}(1,|A|)$. The query is then the conjunction of the $m'$ predicates. Using this procedure, we are able to generate sets of $t$ queries. We vary $t$ from 10 to 2,000. Finally, for a given domain $\mathbb{D}$ and number of queries $t$, we attempt to compute the exact clique number, exact maximum overlap and approximate chromatic number using the aforementioned algorithms. The results are presented for each algorithm in Figures \ref{fig:cliquenumbermontecarlo}, \ref{fig:maxoverlapmontecarlo} and \ref{fig:approxchromaticnumbermontecarlo}, respectively.

\begin{figure*}[!ht]
     \centering
     \begin{subfigure}[b]{0.32\textwidth}
         \centering
         \includegraphics[width=\textwidth]{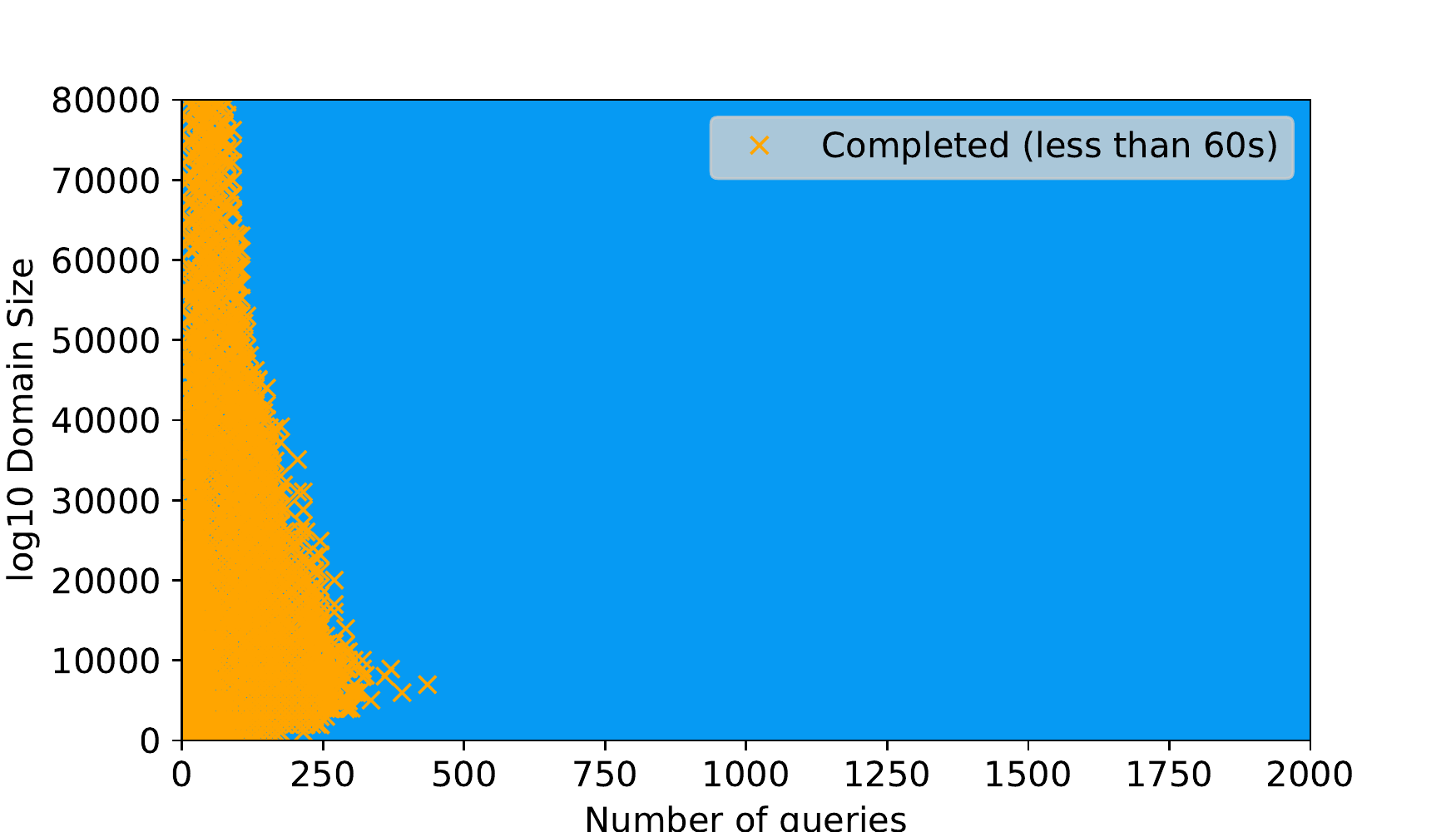}
         \caption{Exact maximum clique}
         \label{fig:cliquenumbermontecarlo}
     \end{subfigure}
     \hfill
     \begin{subfigure}[b]{0.32\textwidth}
         \centering
         \includegraphics[width=\textwidth]{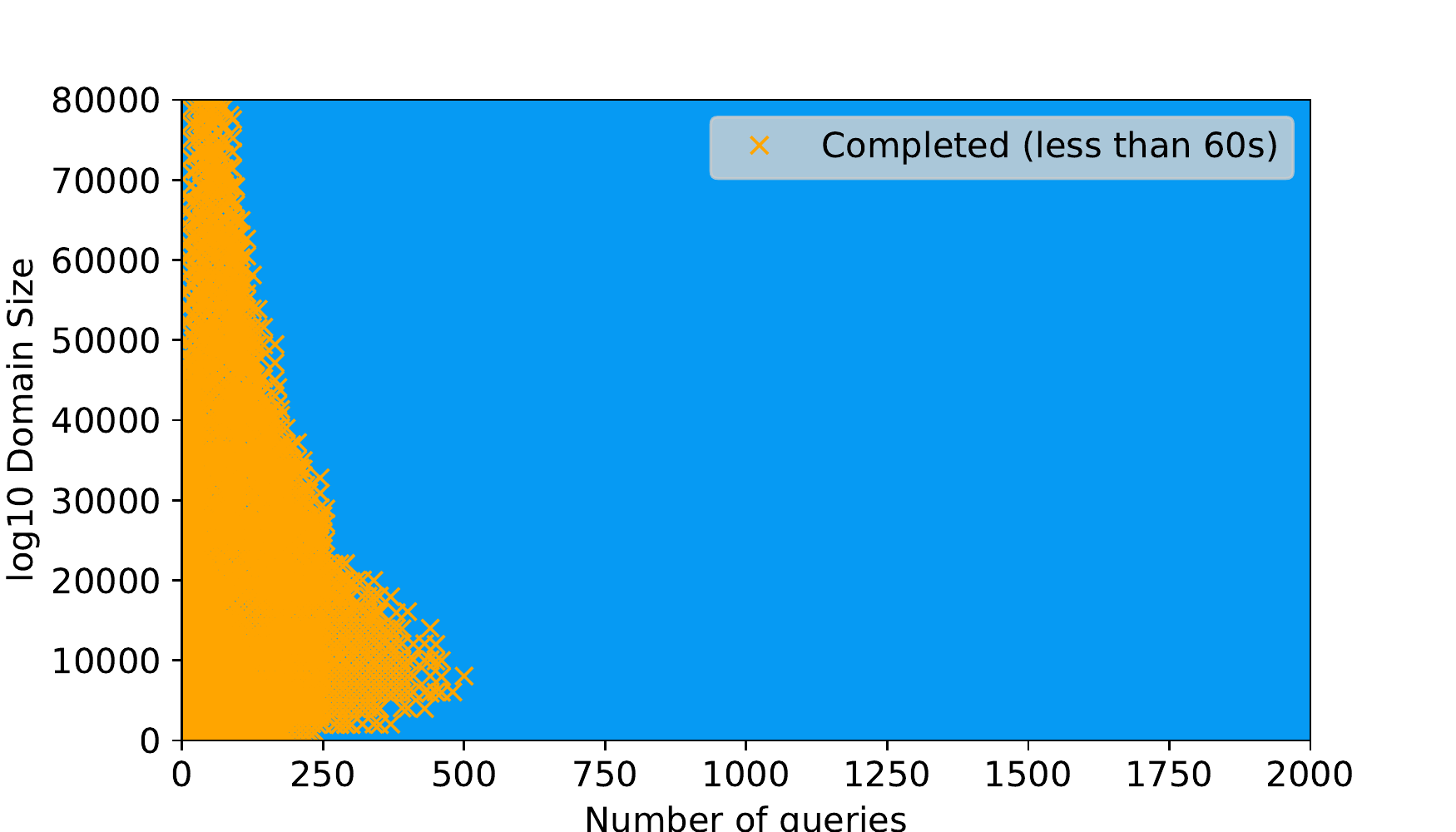}
         \caption{Exact maximum overlap}
         \label{fig:maxoverlapmontecarlo}
     \end{subfigure}
     \hfill
     \begin{subfigure}[b]{0.32\textwidth}
         \centering
         \includegraphics[width=\textwidth]{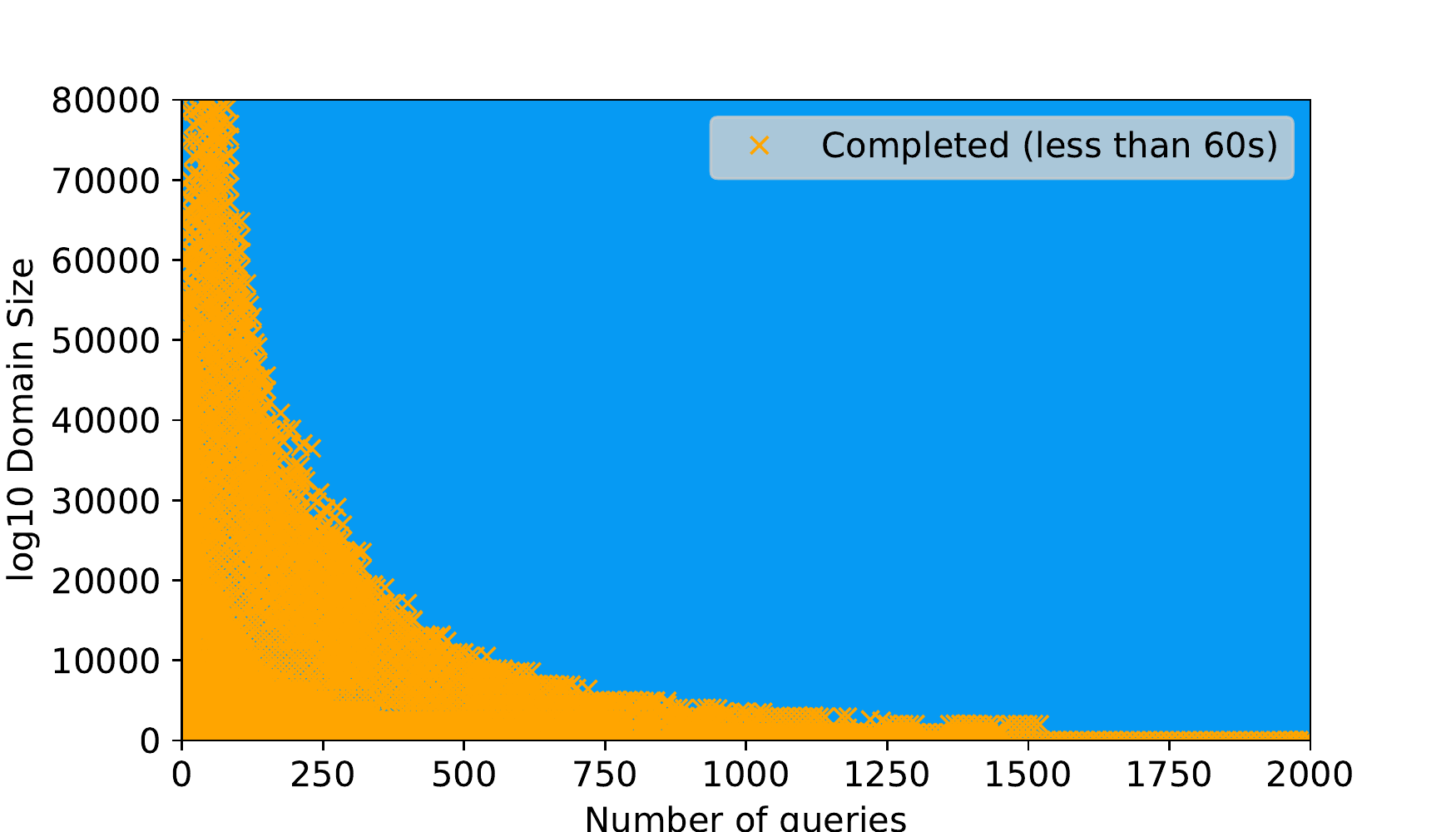}
         \caption{Approximate chromatic number}
         \label{fig:approxchromaticnumbermontecarlo}
     \end{subfigure}
        \caption{Feasible regions for the exact maximum clique, exact maximum overlap and approximate chromatic number algorithms}
        \label{fig:montecarlo}
\end{figure*}

Figures \ref{fig:cliquenumbermontecarlo} and \ref{fig:maxoverlapmontecarlo} show very similar patterns for scaling. This makes sense intuitively, as the algorithm used for computing maximum overlap is based on the algorithm for computing maximum clique. Note that both algorithms time-out for relatively small number of queries on very small domains, i.e., $\log_{10} |\mathbb{D}| \approx 10000$ and number of queries $t \approx 350$. The reason behind this is that due to the query generation process, queries on larger domains are more likely to be disjoint from one another. The peak around $\log_{10} |\mathbb{D}| = 10000$ indicates that the algorithm is most efficient at a certain likelihood of disjointness. This is evident from Figure~\ref{fig:intersection}, where the instances of the maximum clique algorithm are divided into three classes based on the size of the maximum clique. When $\log_{10} |\mathbb{D}| \leq 10000$, there are many instances of maximum cliques of sizes $\ge 10$, and almost no such instances exist for larger domain sizes. Thus, the queries overlap more for smaller domains, and hence the algorithm takes longer to compute the corresponding maximum clique or maximum overlap. 

\begin{figure}[!ht]
\centering
\includegraphics[width=0.4\textwidth]{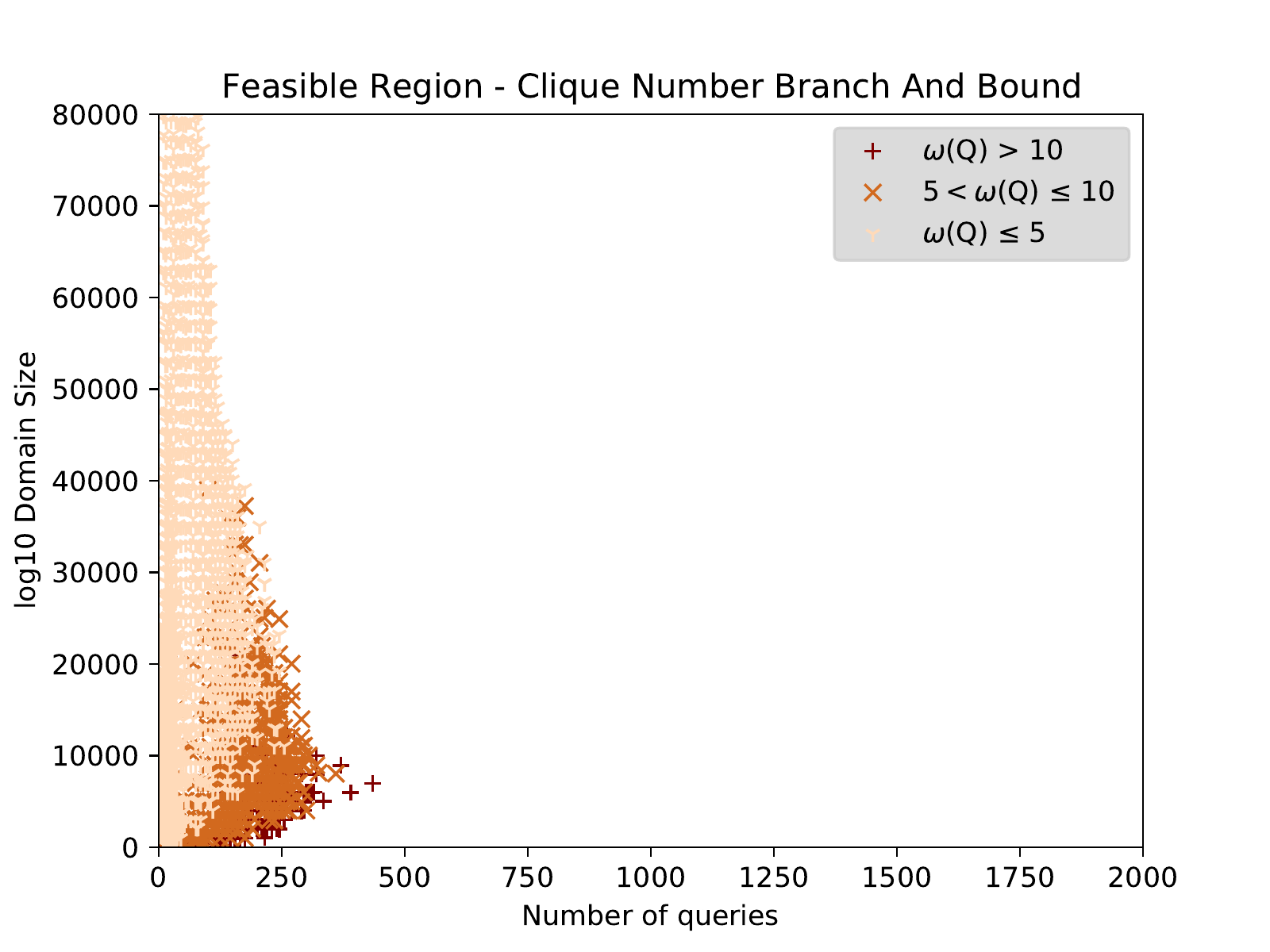}
\caption{Feasible region for different values of the exact clique number $\omega(Q)$}
\label{fig:intersection}
\end{figure}

By contrast, Figure~\ref{fig:approxchromaticnumbermontecarlo} shows a much smoother curve for the feasible region of the approximate chromatic number algorithm. The algorithm is able to handle much larger query sets on smaller domains, and for very small domains is able to handle thousands of queries within the allowed 60-second processing time. The algorithm also runs to completion on almost all cases where the maximum overlap and clique number algorithms run to completion. In Figure~\ref{fig:compare-max-query}, we compare the three algorithms in terms of the maximum number of queries handled as a function of the domain size. As seen from the figure, the approximate chromatic number is able to handle a larger number of queries before being timed out as compared to the other algorithms. This is a significant advantage of the approximate chromatic number algorithm, which we shall return to in the next section. Finally, for very large domains ($\log_{10} |\mathbb{D}| > 20,000$) all three algorithms appear to be limited to between 100 and 200 queries. This may indicate that for larger domains, constructing the query graph itself (common to all three algorithms) dominates the run time. This makes sense, as the construction of the query graph has $\mathcal{O}(m t^2)$ time complexity. Thus, run-time is dominated by the graph algorithms for smaller domains, and by query graph construction for larger domains.

\begin{figure}[!ht]
\centering
\includegraphics[width=0.4\textwidth]{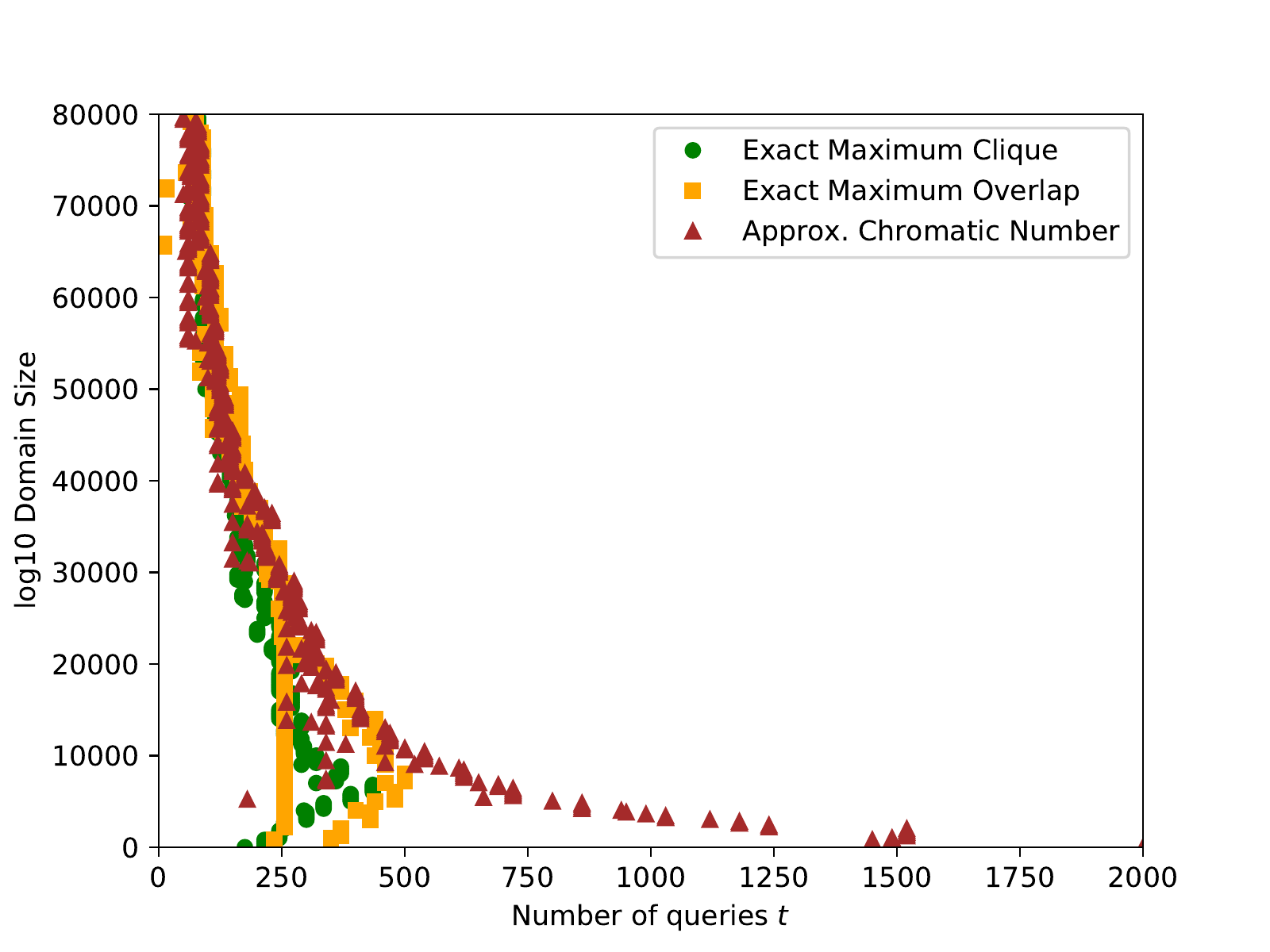}
\caption{The maximum number of queries processed within 60 seconds by the three algorithms, as a function of domain size}
\label{fig:compare-max-query}
\end{figure}

\descr{Other Distributions.} In generating random queries above, we assumed that the following quantities are distributed uniformly:

\begin{enumerate}
\item the number of predicates in a given query;
\item which attributes occur in the predicates in a given query;
\item the number of attribute values occurring in the predicates in a given query; and
\item which attribute values occur in the predicates in a given query.
\end{enumerate}

We generalized our experiments by considering some alternative distributions of the above quantities. For the first and third quantities, we considered the exponential distribution (with different scale parameters). For the second and fourth quantities, we considered the normal distribution (with different standard deviations).

For reasons of computational feasibility, we capped the domain size at $10^{48}$ (instead of $10^{80000}$) and regarded non-termination of an algorithm within 10 seconds (instead of 60 seconds) as a time-out. We focused on the `comfortable query limit', i.e., the maximum number of queries that could with very high probability be processed by an algorithm before time-out.

We varied the distribution for each of the aforementioned four quantities in turn while fixing the uniform distribution for the remaining three. We observed the following discrepancies with the previous experiments:

\begin{itemize}
\item For the approximate chromatic number algorithm, when the numbers of attribute values occurring in predicates were distributed exponentially, the comfortable query limit increased by about 23\%.
\item For the clique number algorithm, when the numbers of predicates occurring in queries were distributed exponentially, the comfortable query limit decreased by about 333\%. Also, when the numbers of attribute values occurring in predicates were distributed exponentially, the limit increased by about 46\%. As the parameter of the distribution increased, the probability of time-out increased modestly. Finally, when attribute values occurred in predicates according to a normal distribution, the limit increased by about 92\%.
\item For the maximum overlap algorithm, when attributes occurred in predicates according to a normal distribution, the comfortable query limit increased by about 20\%. Also, when the numbers of attribute values occurring in predicates were distributed exponentially, the limit increased by about 260\%. As the scale parameter increased, the probability of time-out increased modestly.
\end{itemize}


\subsection{Utility Gain on Random Synthetic Census Queries}
\label{subsec:randomcensusqueries}
In this section, we compare the utility gain obtained via the three graph algorithms as a function of the number of queries, and show that even when the approximate maximum overlap returned by the three algorithms is the same, the approximate chromatic number algorithm has an advantage over the other two in terms of time-outs. For this, we analyze a workload of queries on a census-like data set from \citep[Section 9.2]{zhang2018ektelo} (see also \cite{current-population-survey-data}) and examine the utility gain by taking maximum overlap into account. We assume that each query is allocated the same privacy budget and answered by a homogeneous DP mechanism, e.g., the Gaussian mechanism under $\mu$-GDP (Definition~\ref{def:gauss}), or the Laplace mechanism under $\epsilon$-DP (Definition~\ref{def:laplace}). From Eq.~\ref{eq:util-simple} in Section~\ref{sec:util-gain}, this means that the utility gain (in both cases) is $1 - \tilde{\gamma}/t$, where $\tilde{\gamma}$ is the maximum overlap returned by the algorithm, and $t$ is the number of queries. 
\iffull
By Theorem \ref{thm:arbitrarycomposition}, we know that the true privacy budget usage when considering optimal composition will be given by the maximum overlap of the queries, and approximated by the clique number and chromatic number of the query graph.
\fi

The census data set discussed in \citep{zhang2018ektelo} consists of the following attributes:

\begin{itemize}
    \item \texttt{Income}: 5,000 uniformly sized ranges on the interval (0, 75,0000)
\    \item \texttt{Age}: 5 uniformly sized ranges on the interval (0, 100)
    \item \texttt{Marital status}: 4 discrete values
    \item \texttt{Race}: 7 discrete values
    \item \texttt{Gender}: 2 discrete values
\end{itemize}

This gives a total domain size of $|\mathbb{D}| = 1.4 \times 10^6$. 
\iffull
The authors in \citep{zhang2018ektelo} discuss three query workloads on this census data set. The first two workloads consist solely of queries that cover the entire domain, and hence none of the queries will compose in parallel. The third workload however is very 
\fi
nicely suited to our problem. The workload consists of all queries of the form $(\texttt{Income} \in (0, i), \texttt{Age} == a, \texttt{Marital status} == m, \texttt{Race} == r, \texttt{Gender} == g)$, where $(0, i)$ is an income range and $a, m, r$ and $g$ are either single elements from the field's domain or all elements from the field's domain. 

\iffull
There are 3,600,000 queries in this workload, which is far larger than our test sets in Section \ref{subsec:scalingexperiments}. To work around this, we design an approach for randomly sampling a set of queries from the workload. Firstly, we assume the possible values of $i, a, m, r$ and $g$ are uniformly distributed, and all independent from one another. Thus, to generate
a random query, we simply sample random values for $i, a, m, r$ and $g$. We can use this process to generate a random set of queries of size $t$. For each set of queries, we compute the maximum overlap, maximum clique and approximate chromatic number to find the (approximate) privacy budget usage. 
In our experiments we vary $t$ from 25 to 2,000.

During our experiments, we found that the utility gained varied significantly depending on the query sets generated. To work around this randomness, for each value of $t$, we generate 30 different query sets instead of 1, and report the mean of the maximum overlap, clique number and approximate chromatic number. 
\fi
These results are given in Figure~\ref{fig:censusmeanbudget}. The figure shows that the utility gain remains within the 85--95\% bracket for all three algorithms, and around 95\% for most of the queries. This utility gain is huge, but easily explained by the highly parallel nature of the workload. Importantly, Figure \ref{fig:censusmeanbudget} indicates that for this data set and method of sampling queries, the gap between the approximate chromatic number and the maximum overlap is very small. This contrasts with the theoretical results in Section \ref{sec:graphtheory}, where the gap between the chromatic number and maximum clique can be arbitrarily large (see \citep{mycielski1955coloriage}).

\begin{figure}[h]
\centering
\includegraphics[width=0.4\textwidth]{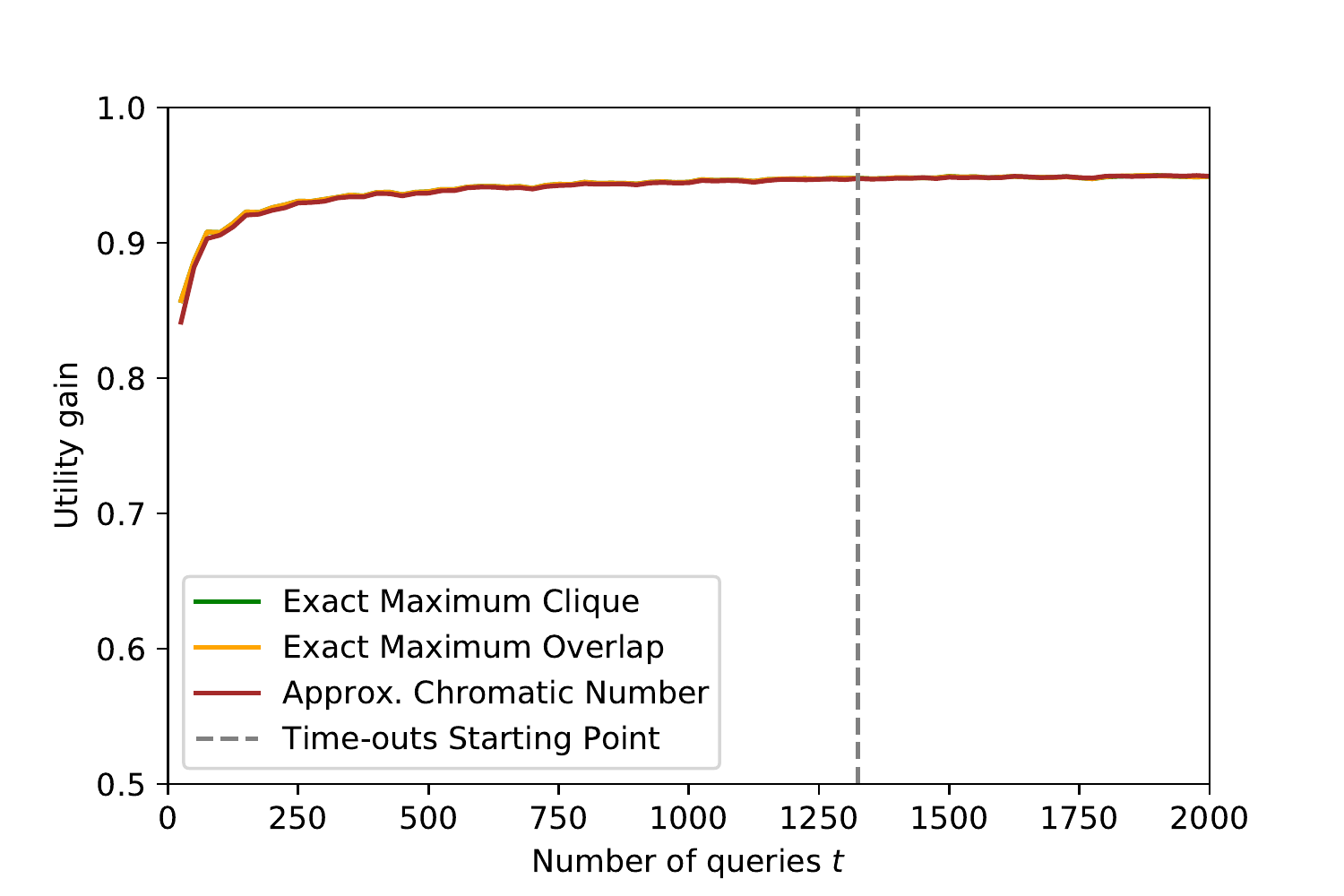}
\caption{Utility gain as a function of number of queries. Maximum overlap and maximum clique always gave the same results, as well as approximate chromatic number in most cases, and hence the lines overlap.}
\label{fig:censusmeanbudget}
\end{figure}

However, there is some difference in the algorithms in terms of execution time. The dotted vertical line in the figure indicates the starting point (as a function of the number of queries) where the maximum overlap and clique number algorithms start to time out (i.e., take more than 60 seconds to execute). This is illustrated in Figure~\ref{fig:censuspercomplete}, where we show the percentage of the 30 randomly sampled query sets completed within 60 seconds for each value of $t$. While the approximate chromatic number algorithm always gives an output within 60 seconds, at $t$ $=$ 1,350 we begin to see that the other two algorithms start to time out, with the percentage of time-outs increasing as $t$ grows.



\begin{figure}[h]
\centering
\includegraphics[width=0.4\textwidth]{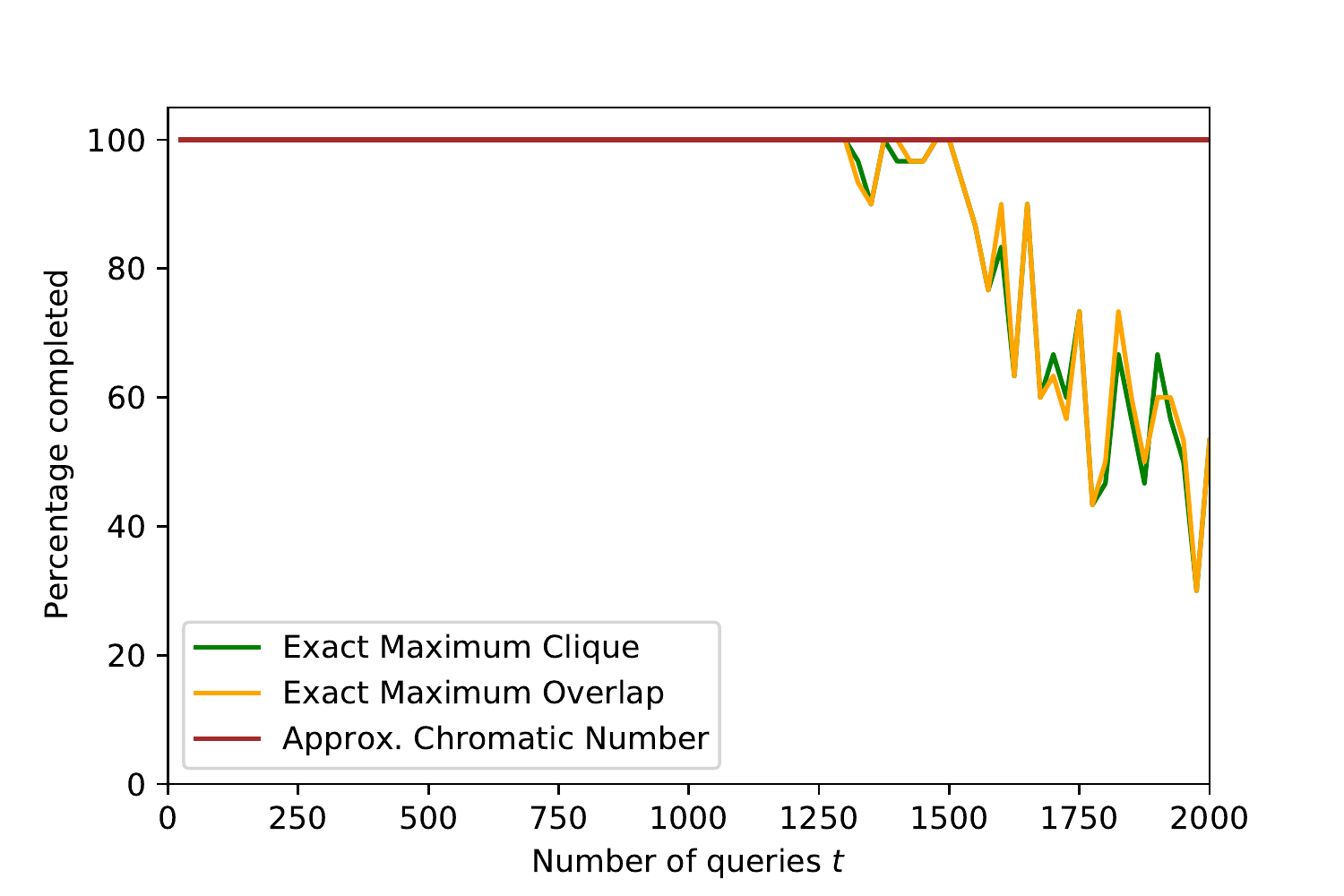}
\caption{The percentage of the query sets for which the three algorithms output a result before time-out.}
\label{fig:censuspercomplete}
\end{figure}

\iffull
In Figure~\ref{fig:censustimetaken}, we illustrate the (mean) time taken by the algorithms as a function of $t$ for all queries executed within 60s. As we see, with increasing $t$, the clique number and maximum overlap algorithms take longer, whereas the performance of the chromatic number algorithm degrades more gracefully.

\begin{figure}[h]
\centering
\includegraphics[width=0.45\textwidth]{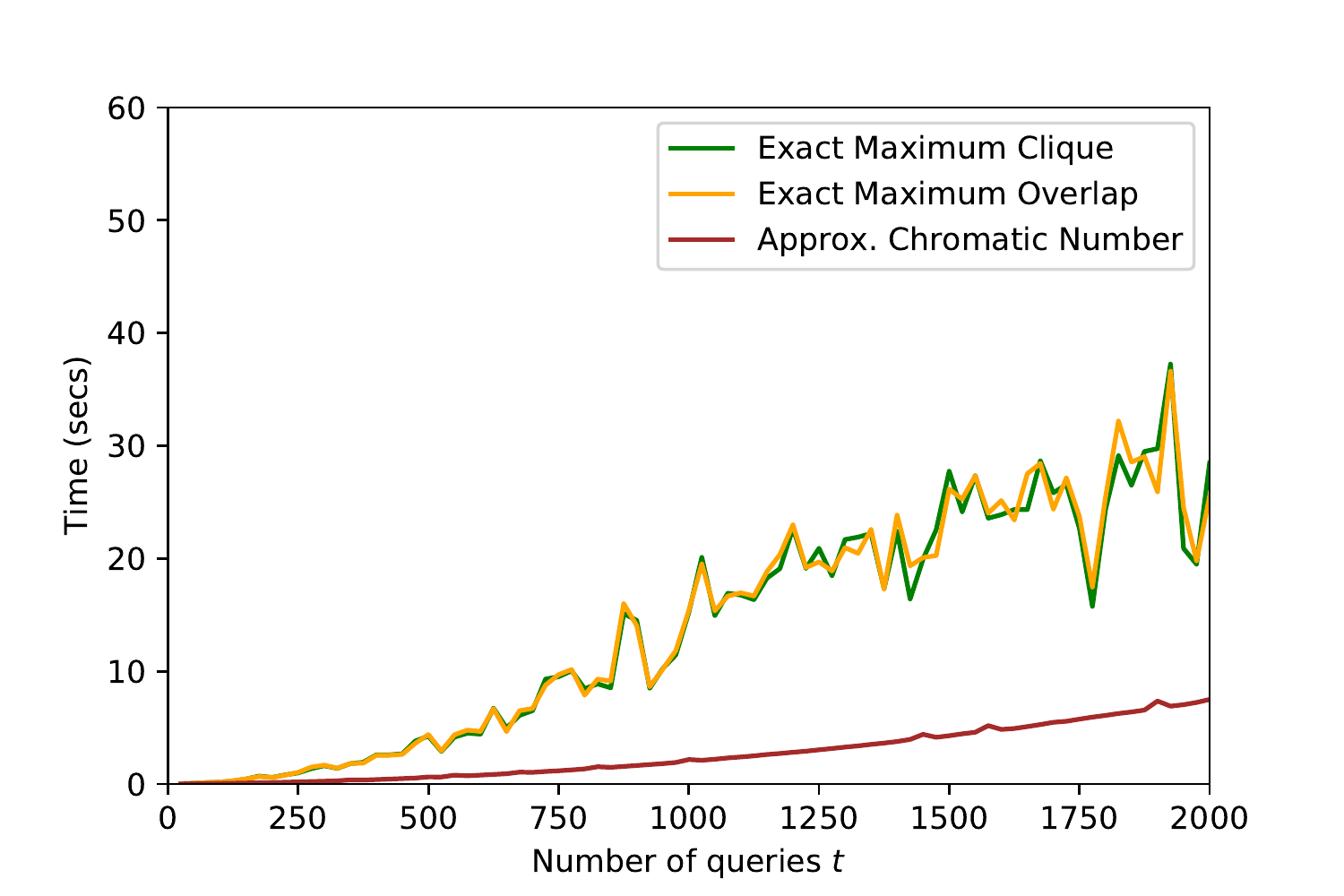}
\caption{The average time taken by each algorithm on query sets completed before time-out.}
\label{fig:censustimetaken}
\end{figure}
\fi




\descr{Other Distributions.} 
Just as was done for the scalability experiments reported in Section \ref{subsec:scalingexperiments}, we generalized the experiments for utility gain by considering non-uniform distributions (in turn) for the four salient quantities controlling the random generation of queries identified in that section. 
\iffull
Since the clique number and maximum overlap algorithms did not scale well for a larger set of queries, i.e., timing out around 50 to 80 queries as shown in Section~\ref{subsec:scalingexperiments}, we focus on the results for the approximate chromatic number algorithm for query sets of size up to $1000$. For each query set size, we repeated the experiment three times, and report the results in Figure~\ref{fig:budgetdistributionschromatic}. For three of the four quantities, the utility gain through the uniform distribution is either comparable or better than the other distributions. The utility gain is highest for the uniform distribution on the number of predicates selected. The uniform distribution in this case can select a larger number of predicates, thus making more queries overlap. On the other hand the utility gain for the uniform distribution on the number of values taken by a predicate is the lowest. This is again explainable, as the exponential distribution on the number of values taken per predicate means that the resulting queries are more likely to overlap. In all cases, we observe considerable utility gain. For completeness, we show the results for the maximum clique and maximum overlap algorithms for a small number of queries in Appendix~\ref{app:mo-cliq-util-gain}. 
\fi

\iffull
\begin{figure*}[!ht]
     \centering
     \begin{subfigure}[b]{0.45\textwidth}
         \centering
         \includegraphics[width=\textwidth]{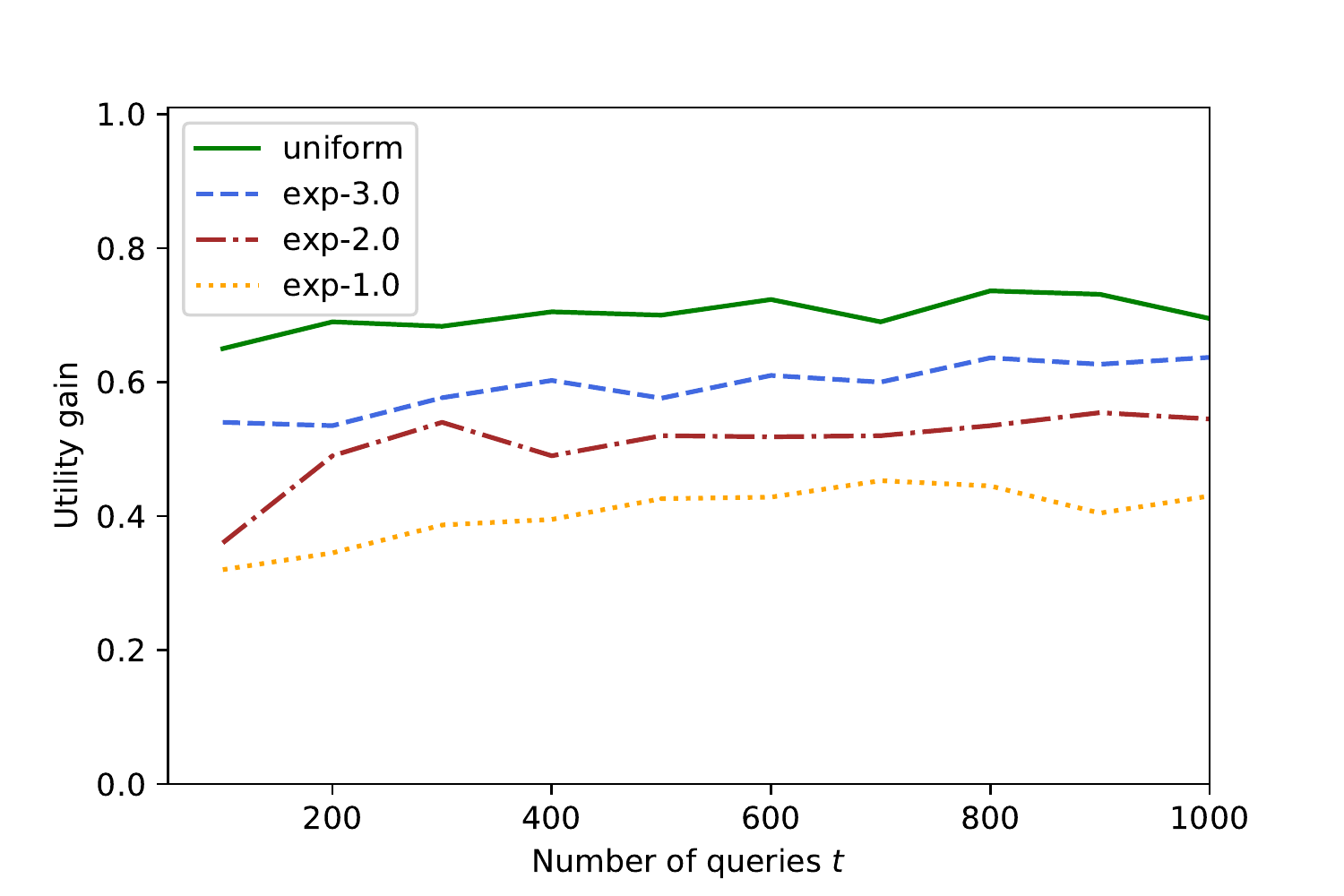}
         \caption{Influence of distribution on the number of predicates}
         \label{fig:chromaticnumpredicates}
     \end{subfigure}
     \begin{subfigure}[b]{0.45\textwidth}
         \centering
         \includegraphics[width=\textwidth]{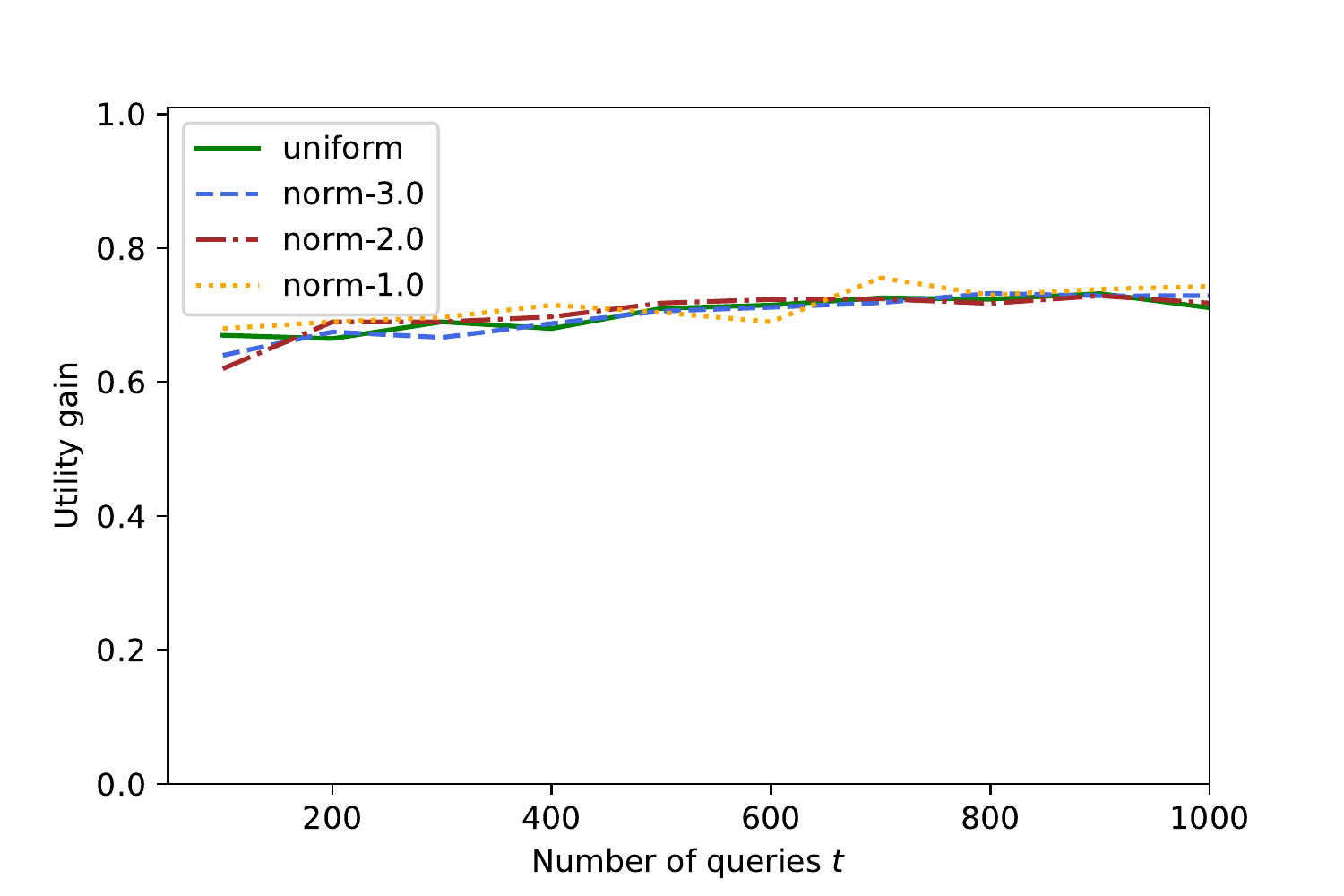}
         \caption{Influence of distribution on the attributes}
         \label{fig:chromaticfields}
     \end{subfigure}
     \begin{subfigure}[b]{0.45\textwidth}
         \centering
        \includegraphics[width=\textwidth]{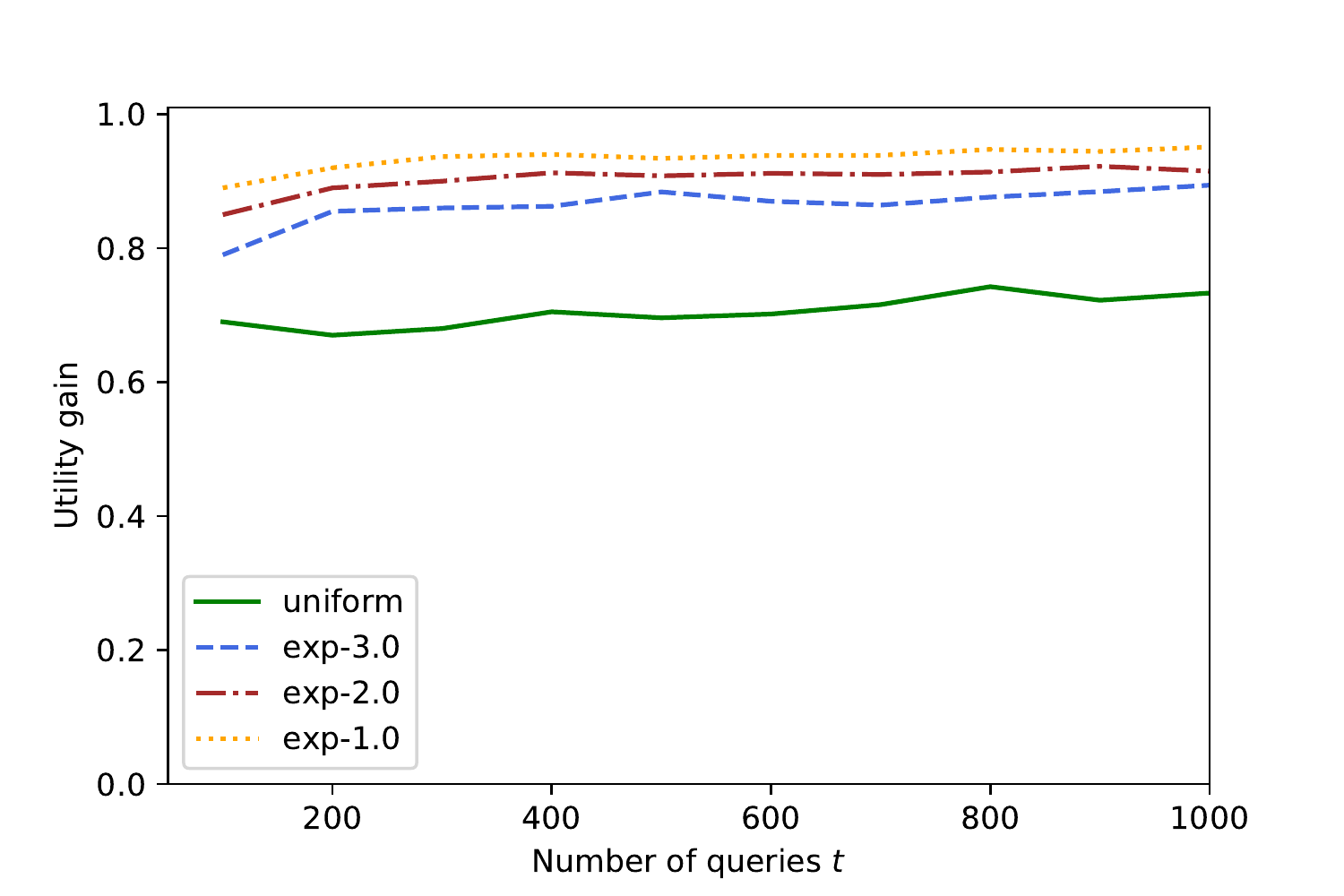}
         \caption{Influence of distribution on the number of values per predicate}
         \label{fig:chromaticnumnumvalues}
     \end{subfigure}
     \begin{subfigure}[b]{0.45\textwidth}
         \centering
         \includegraphics[width=\textwidth]{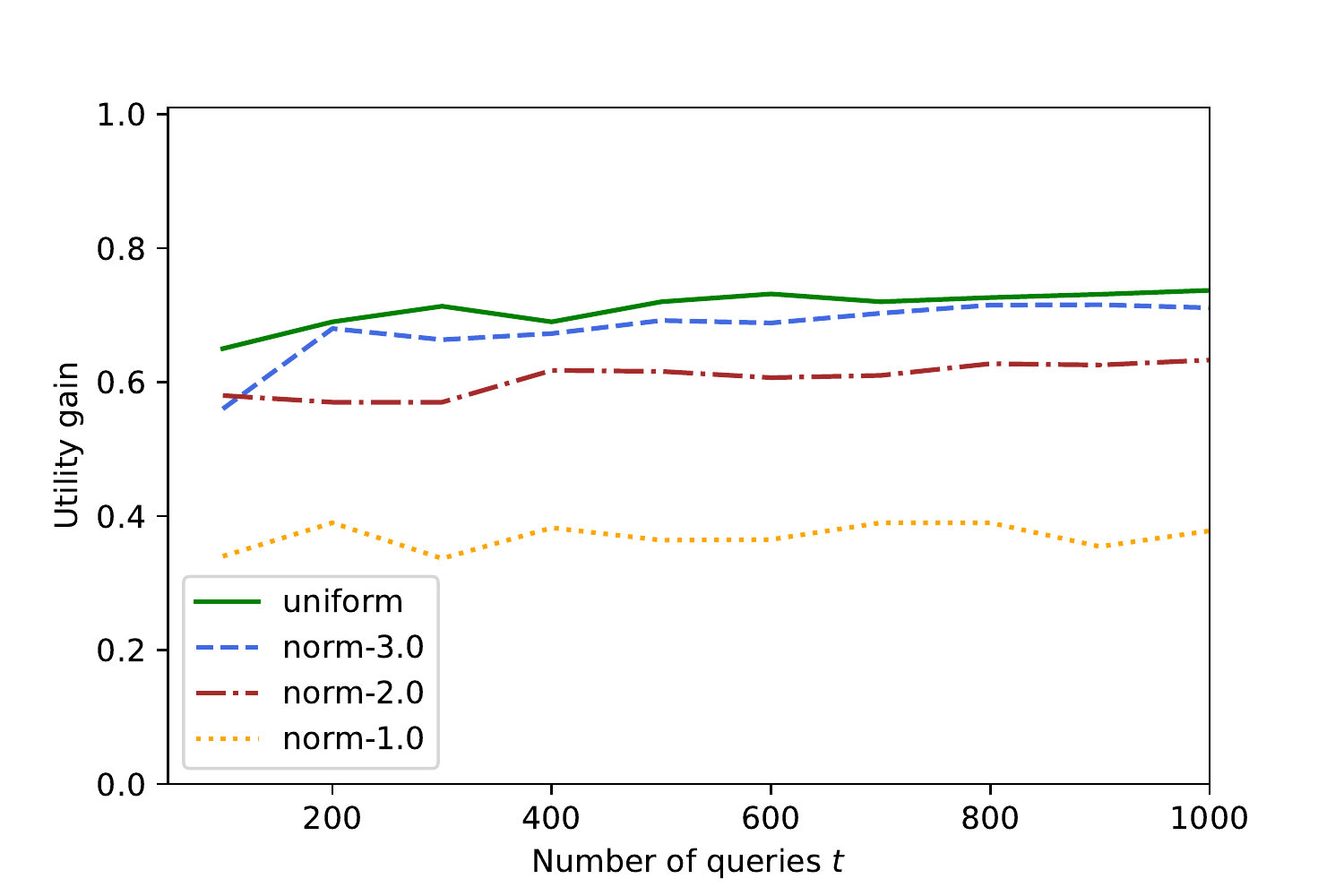}
         \caption{Influence of distribution on the attribute values}
         \label{fig:chromaticnumvalues}
     \end{subfigure}
        \caption{Utility gain versus number of queries, for the \textbf{approximate chromatic number} algorithm, for different distributions on the four quantities (random variables) used in generating queries from Section~\ref{subsec:scalingexperiments}. The label `uniform' refers to the uniform distribution; `exp' refers to the exponential distribution (with three different scale parameters); and `norm' refers to the normal distribution (with three different standard deviations).}
        \label{fig:budgetdistributionschromatic}
\end{figure*}
\fi


\subsection{Real Census Queries Data Set}
\label{sub:absqueryresults}

We also analyze a data set of queries on 
the Australian national census data, logged by Australia's national statistics agency, the Australian Bureau of Statistics.
The data set contains nine separate workloads of queries for a census-like data set. The census-like data set has a domain of size $|\mathbb{D}| \approx 6.8 \times 10^{28}$. We analyze the benefit of running our approach on the query sets using the utility gain metric introduced in Section \ref{sub:max-overlap}. Recall that the utility gain was defined as $U = 1 - \gamma / t$, where $t$ is the number of queries. For this experiment, we compute the maximum overlap and utility gain for each set of queries. We also compute these values for the workload of all queries combined together. The results can be seen in Table \ref{table:absutilityanalysis}.

\begin{table}
\centering
\begin{tabular}{ c|c|c|c|c|c } 
\multirow{2}{*}{Data set} & \multirow{2}{*}{$t$} & \multirow{2}{*}{$\gamma$} & \multirow{2}{*}{Utility gain} & \multicolumn{2}{c}{Average $l_1$ Error}\\
\cline{5-6}
&&&& Seq. & Opt.\\
    \hline
    1 & 9 & 6 & 0.333 & 2.392 & 1.9488\\
    2 & 120 & 107 & 0.108 & 8.747 & 8.254 \\
    3 & 2 & 2 & 0.000 & 1.133 & 1.133\\
    4 & 267 & 216 & 0.191 & 13.040 & 11.734\\
    5 & 54 & 34 & 0.370 & 5.850 & 4.657\\
    6 & 68 & 55 & 0.191 & 6.573 & 5.921\\
    7 & 41 & 17 & 0.585 & 5.116 & 3.286\\
    8 & 38 & 20 & 0.474 & 4.912 & 3.568\\
    9 & 284 & 208 & 0.268 & 13.446 & 11.516\\
    \hline
    Combined & 883 & 563 & 0.362 & 23.709 & 18.940\\
\end{tabular}
\caption{Utility gain on the real census query data set and the comparison of average absolute error through $\mu$-GDP with $\mu = 1$ under sequential composition (Seq.) versus optimal composition (Opt.)}
\label{table:absutilityanalysis}
\end{table}

With the exception of data set 3, improvements in utility range from 10.8\% to 58.5\%. We found that on all data sets (including the combined data set), there was no gap between the approximate chromatic number and the true maximum overlap (and hence $\tilde \chi(Q) = \omega(Q) = \gamma(Q)$). The overall utility gain for the combined data set was 36.2\%. 

We also present the results using the more familiar average $l_1$ error metric (Eq.~\ref{eq:l1}) in Table~\ref{table:absutilityanalysis}. We choose the Gaussian mechanism which is $\mu$-GDP private. We set the overall budget to be $\mu = 1$. By sequential composition of $\mu$-GDP mechanisms, this means that each of the $t$ queries in $Q$ is allocated a budget of $\mu' = 1/\sqrt{t}$. Using optimal composition, we allocate each query a budget of $\mu' = 1/\sqrt{\gamma}$. The results for the two cases are displayed in columns labelled `Seq.' and `Opt.', respectively. Notice that with optimal parallel composition there is a significant decrease in absolute error, resulting in an overall error of less than 4.7; a significant improvement for sensitivity-1 queries.

\section{Related Work}
\iffull
The parallel composition theorem, which is the main theme of this paper, was proven for $\epsilon$-differential privacy by~\cite{mcsherry2009privacy}, who commented on its importance in practical privacy platforms. As compared with sequential composition, researchers have paid little attention to the computational aspects of parallel composition. When it is mentioned at all, the parallel composition theorem is usually applied as part of an analysis of the privacy loss associated with the execution of a specific data release mechanism (see, e.g., \citep{asghar2019differentially, diff-gen}). Often the operation of the mechanism involves the issuing of `measurement' queries (as opposed to queries asked directly by the user), and the ways in which these queries overlap is predetermined so that computational issues do not arise.
\fi

The work most closely related to our own is contained in \cite{xiao-np} and \cite{inan-sensitivity}. The former work was the first to prove that computing the $l_1$-sensitivity of a set of queries is NP-hard. As we have discussed, the notion of maximum overlap under basic sequential composition is equivalent to finding the $l_1$-sensitivity of a set of queries. However, our treatment of the maximum overlap in terms of its $f$-differential privacy characterization is much broader, and includes other mechanisms such as the Gaussian mechanism for which the $l_1$-sensitivity result does not apply. The work from~\cite{inan-sensitivity} formulates the problem of computing $l_1$-sensitivity of a set of  statistical range queries as a graph problem. They then compute $l_1$-sensitivity via exact maximum clique algorithms. As mentioned before, we are bound to use exact maximum clique algorithms, and not their approximate counterparts, because the result may underestimate the privacy budget, hence potentially causing privacy leakage. We have additionally linked finding the chromatic number of the graph to maximum overlap, with the advantage that the approximate chromatic number algorithms never underestimate the privacy budget. As a result, we are able to run the algorithm for a much larger number of queries and domain sizes than was possible in~\citep{inan-sensitivity}. We also remark that our results hold for the set of predicate queries, a large class of queries that properly includes the statistical range queries considered in \citep{xiao-np} and \citep{inan-sensitivity}.

Zhang et al.~\cite{zhang2018ektelo} introduce approaches for computing a reduced workload matrix of queries. Their workload-based `partition selection' operator directly takes advantage of parallel composition, and this results in improved accuracy. The authors also devise ways of exploiting the structure of range queries, and give a modified version of the Multiplicative Weights Exponential Mechanism (MWEM) \citep{hardt2010multiplicative} that selects groups of queries that are pairwise disjoint. Our approach is able to scale to much larger data domains.

The framework of personalized differential privacy \citep{ebadi2015differential} is based on a notion that is closely related to our notion of the maximum overlap of a set of queries. In this framework each individual \emph{actually} in the data set is assigned a separate maximum allowed privacy loss, and every query that accesses the individual's data increases the individual's privacy loss. In contrast, we consider the domain of all individuals that \emph{could} be in the data set and, crucially, discuss the computational aspects of determining how queries cover the data domain.

The designers of the High-Dimensional Matrix Mechanism (HDMM) \citep{mckenna} consider the problem of simultaneously maximising accuracy of query answers and minimising privacy loss for a workload of predicate queries of the same type as that considered by us. However, whereas we focus on this class of queries for the purpose of reducing time complexity, the designers of HDMM do so to reduce space complexity. Specifically, workloads of such queries can be compactly represented through use of the Kronecker product. However, experimental results indicate that the run-time of HDMM scales with the size of the data domain (as opposed to the number of attributes). Thus, in practice HDMM can handle only small domains.

HDMM is a well-known example of a `workload-aware' differentially private mechanism. Such mechanisms execute an optimisation routine (such as a least-squares regression) in order to maximize the expected accuracy of the answers to the queries in a given workload of queries. While it is likely that such mechanisms make use of parallel composition, they do so only implicitly. Indeed, to date, no work has been conducted to determine the extent to which workload-aware mechanisms exploit parallel composition. Our approach addresses the problem of making optimal use of parallel composition directly, thereby avoiding the significant computational overheads associated with most of the optimisation routines used by workload-aware mechanisms.

McKenna et al.~\cite{mckenna2019graphical} use probabilistic graphical models (PGMs) \citep{koller2009probabilistic} to address the problem of inference in high-dimensional data sets. Rather than building an explicit probability vector over all elements of the data domain, the use of PGMs allows for a compact, implicit representation whose size scales with the number of attributes. Our graph-based framework has similar benefits, but is intended for the measurement of privacy loss, rather than the optimisation of inference. Future work in this area could involve combining the two frameworks.

\iffull
Our work is different from work on differential privacy for graph data sets, which often treats the problem of preserving the privacy of users of social media \citep{nissim2007smooth, task2012guide, kasiviswanathan2013analyzing, xiao2014differentially}. In our problem, data sets are represented as tables, not as graphs. It might be possible to extend our framework to allow for analysis of queries on graph data sets.
\fi

\section{Limitations}
We list a few shortcomings of our work which, if addressed, could further improve our work.
\begin{itemize}
    \item Our primary use case is the online setting where the data custodian answers queries on the fly. Our approach is to regenerate the query graph whenever a new batch of queries arrives. However, an approach based on a dynamic query graph~\cite{das2019incremental} could improve the overall query processing time, because it would eliminate the need to regenerate the query graph from scratch. 
    \item A possible method of reducing the size of the query graph is to represent complex queries which have a pre-determined, regular structure as single nodes in the graph. For instance, an SQL `GROUP BY' query across different attribute values of a single attribute $A$ (i.e., a histogram query) has such a structure. Such a query might be representable as a single node, instead of $|A|$ nodes, in the query graph. This idea leads to the more general question of whether it is possible to efficiently pre-process the query graph to reduce its size before executing the graph algorithms.
    \item We evaluated our approach on real-world workloads of queries in Section~\ref{sub:absqueryresults}. Unfortunately, there is a lack of publicly available such workloads. The TPC-H data set~\cite{tpc-h} for database performance benchmarking contains samples of real-world queries over multiple data sets, but not workloads of queries on single data sets.
\end{itemize}

\section{Conclusion and Future Work}
We have shown that making the optimal use of parallel composition amounts to computing the maximum overlap of a set of queries. Although computing the maximum overlap is NP-hard, it is possible to approximate this quantity using well-known graph algorithms, e.g., by using efficient approximate algorithms for the chromatic number of the query 
\iffull
graph. Although in theory the approximation error can be arbitrarily large, our experiments have shown that the error is frequently very small, and hence that the approximate algorithms are useful in practice. Our experiments have also shown that application of our approach leads to significant gains in utility when many of the queries are disjoint from one another. 
\fi
It would be interesting to broaden the scope of our approach to include additional classes of queries. Range queries form one particularly interesting class of queries. Each predicate in a range query can be represented as an interval, so that the query can be regarded as a box (hyper-rectangle) in $m$-dimensional space (where $m$ is the number of attributes). We remark that the maximum overlap problem for range queries is equivalent to the maximum depth problem in computational geometry \citep{chan2013klee}. 
Another interesting direction is to investigate whether more could be squeezed out of the maximum overlap problem. For instance, we have defined an overlap as a binary function: two queries overlap if they intersect on at least one row in the domain. Is it possible to obtain an even tighter privacy analysis by considering the amount (the number of rows in the domain) by which queries overlap? We leave this as an open question. 

\section*{Acknowledgements}
We thank our shepherd Catuscia Palamidessi and the anonymous reviewers for their suggestions which have helped us significantly improve this paper. This research received no specific grant from any funding agency in the public, commercial, or not-for-profit sectors.

\bibliographystyle{plain}
\bibliography{ref}

\appendix

\section{Hypergraphs and Maximum Overlap}
\label{app:hypergraph}
A \emph{hypergraph} is a generalisation of the concept of a graph, where the elements of $E$ are non-empty subsets of any cardinality of $V$ (i.e., not simply two-element subsets)~\citep{diestel-gt}. In the context of hypergraphs, an element of $E$ shall be called a hyperedge. A simple brute-force algorithm for finding $\gamma(Q)$ would check all $2^t$ subsets of $Q$, check whether each subset has a non-empty coverage, and report the largest subset with a non-empty coverage. Performing such a brute-force search would allow us to construct the following hypergraph:

\descr{Overlap hypergraph.} Given a set of queries $Q \coloneqq \{q_1, q_2, \ldots, q_t\}$, their \emph{overlap hypergraph}, denoted $\mathcal{H}(Q)$, is defined as follows: each query is a vertex (i.e., $V = Q$), and hyperedges are subsets of queries with a non-empty coverage, i.e.,
\[
E \coloneqq \{ Q' \subseteq Q : C_{Q'}(\mathbb{D}) \neq \emptyset\}.
\]

\descr{Down-closed hypergraph.} A hypergraph $\mathcal{H} = (V, E)$ is \emph{down-closed} if $e_1 \in E$ and $e_2 \subseteq e_1$ implies $e_2 \in E$ (i.e., every subset of a hyperedge is also a hyperedge) \citep{down-closed}.

\begin{proposition}
$\mathcal{H}(Q)$ is down-closed.
\end{proposition}
\begin{proof}
If a set of queries $Q' \subseteq Q$ has a non-empty coverage $C_{Q'}(\mathbb{D}) \neq \emptyset$, then every subset of $Q'' \subseteq Q'$ also has $C_{Q''}(\mathbb{D}) \neq \emptyset$. Thus, every subset of a hyperedge in $\mathcal{H}(Q)$ is also a hyperedge in $\mathcal{H}(Q)$.
\end{proof}

\descr{Rank.} The \emph{rank} $r(\mathcal{H})$ of a hypergraph $\mathcal{H} \coloneqq (V, E)$ is the cardinality of the largest hyperedge, i.e., $r(\mathcal{H}) \coloneqq \max_{e \in E} |E|$ \citep{voloshin2009introduction}.
\begin{proposition}
\label{prop:rankvsgamma}
Given a set of queries $Q$, one has $\gamma(Q) = r(\mathcal{H}(Q))$.
\end{proposition}
\begin{proof}
For the overlap hypergraph $\mathcal{H}(Q)$, the edges are defined by $E \coloneqq \{Q' \subseteq Q : C_{Q'}(\mathbb{D}) \neq \emptyset\}$. It follows that $\max_{e \in E} |E| = \max_{Q' \subseteq Q} \{ |Q'| : C_{Q'}(\mathbb{D}) \neq \emptyset \} = \gamma(Q)$.
\end{proof}
Thus, computing $\gamma(Q)$ exactly translates to finding the rank of the overlap hypergraph. However, the $\mathcal{O}(2^t)$ running time means it is not very useful in practice. Therefore, we focus on the query graph instead to obtain a lower bound.

\iffull
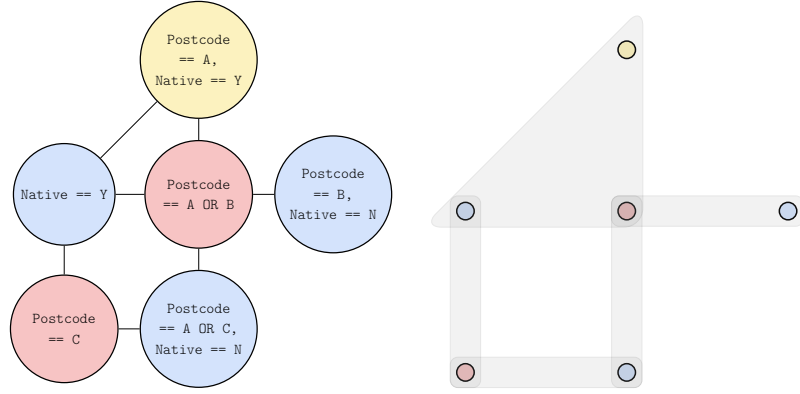
\begin{figure*}[!ht]
\centering
\resizebox{0.3\textwidth}{!}{%
\begin{tikzpicture}
\definecolor{yellow1}{RGB}{250, 240, 190}
\definecolor{blue1}{RGB}{210, 225, 250}
\definecolor{red1}{RGB}{245, 195, 195}

\begin{scope}[every node/.style={circle,thick,draw}]
    \node[align=center, text width=2cm, fill=yellow1] (A) at (3,6) {\texttt{Postcode == A, Native == Y}};
    \node[align=center, text width=2cm, fill=red1] (B) at (3,3) {\texttt{Postcode == A OR B}};
    \node[align=center, text width=2cm, fill=blue1] (C) at (3,0) {\texttt{Postcode == A OR C, Native == N}};
    \node[align=center, text width=2cm, fill=blue1] (D) at (0,3) {\texttt{Native == Y}};
    \node[align=center, text width=2cm, fill=red1] (E) at (0,0) {\texttt{Postcode == C}};
    \node[align=center, text width=2cm, fill=blue1] (F) at (6,3) {\texttt{Postcode == B, Native == N}};
\end{scope}

    \draw (A) -- (B);
    \draw (A) -- (D);
    \draw (B) -- (C);
    \draw (B) -- (D);
    \draw (B) -- (F);
    \draw (D) -- (E);
    \draw (C) -- (E);
    
\end{tikzpicture}
}
\;
\resizebox{0.3\textwidth}{!}{%
\begin{tikzpicture}
[
    he/.style={draw, rounded corners, inner sep = 3pt},   
]
\definecolor{yellow1}{RGB}{250, 240, 190}
\definecolor{blue1}{RGB}{210, 225, 250}
\definecolor{red1}{RGB}{245, 195, 195}

\node[style={circle,thick,draw}, fill=yellow1] (A) at (3,6) {};
\node[style={circle,thick,draw}, fill=red1] (B) at (3,3) {};
\node[style={circle,thick,draw}, fill=blue1] (C) at (3,0) {};
\node[style={circle,thick,draw}, fill=blue1] (D) at (0,3) {};
\node[style={circle,thick,draw}, fill=red1] (E) at (0,0) {};
\node[style={circle,thick,draw}, fill=blue1] (F) at (6,3) {};
    
\node[he, fit=(B) (F), fill = gray, opacity = 0.1] {};
\node[fit=(A) (B) (D)] (fd) {};
\draw [rounded corners=10pt, fill = gray, opacity = 0.1] ($(fd.north east)+(0.0,0.5)$) -- ($(fd.south east)+(0,0)$) -- ($(fd.south west)+(-0.5,0)$)--cycle;

\node[he, fit=(B) (C), fill = gray, opacity = 0.1] {};
\node[he, fit=(D) (E), fill = gray, opacity = 0.1] {};
\node[he, fit=(C) (E), fill = gray, opacity = 0.1] {};

\end{tikzpicture}
}
\caption{The query graph and the overlap hypergraph of the set of queries from Example~\ref{ex:query-graph}.}
\label{fig:query-graph}
\end{figure*}
\fi

\begin{example}
\label{ex:query-graph}
Consider the set $Q$ of six queries:
\begin{align*}
    q_1 &: \texttt{Postcode == A, Native == Y} \\
    q_2 &: \texttt{Postcode == A OR B} \\
    q_3 &: \texttt{Postcode == A OR C, Native == N} \\
    q_4 &: \texttt{Native == Y} \\
    q_5 &: \texttt{Postcode == C} \\
    q_6 &: \texttt{Postcode == B, Native == N}
\end{align*}
Then the query graph $\mathcal{G}(Q)$ and overlap hypergraph $\mathcal{H}(Q)$ are shown in Figure~\ref{fig:query-graph}. Note that in the hypergraph shown in the figure, every subset of the hyperedge shaded via a triangle is also a hyperedge. These hyperedges are not explicitly illustrated in the figure.
\qed
\end{example}

\section{Proofs}
\label{app:proofs}

\iffull
\descr{Proof of Proposition~\ref{prop:trivial}.}
For part (1), assume that $\phi$ is not a contradiction. Then for at least one $a \in A$, we have $\phi(a) = 1$. By definition, for this $a$ we have $I(a) = 1$. Thus, $I$ and $\phi$ overlap. On the other hand if $\phi$ is a contradiction, then 
\[
C_{\phi}(A) \cap C_{I}(A) = C_{I^c}(A) \cap A = \emptyset \cap A = \emptyset.
\]
Therefore, in this case $I$ and $\phi$ are disjoint. For part (2), we have
\[
C_{I^c}(A) \cap C_{\phi}(A) = \emptyset \cap C_{\phi}(A) = \emptyset.
\]
Thus, $I^c$ and $\phi$ are disjoint.\qed
\fi

\descr{Proof of Proposition~\ref{prop:disjoint}.}
First consider that $q_1$ and $q_2$ are disjoint, and assume to the contrary that for all $i$, $\phi_{i, 1}$ and $\phi_{i, 2}$ overlap on $A_i$. Let $x$ be a row whose $i$\textsuperscript{th} coordinate is a member of the set $C_{\phi_{i, 1}}(A_i) \cap C_{\phi_{i,2}}(A_i)$, which by assumption is non-empty. Then $x$ satisfies both $q_1$ and $q_2$, contradicting the fact that they are disjoint.

Next assume that for some $i$, $\phi_{i, 1}$ and $\phi_{i, 2}$ are disjoint on $A_i$. Let $x \in \mathbb{D}$, then its $i$\textsuperscript{th} coordinate can not simultaneously satisfy $\phi_{i, 1}$ and $\phi_{i, 2}$. If $x_i$ satisfies neither, then it does not satisfy both $q_1$ and $q_2$. If $x_i$ satisfies $\phi_{i, 1}$, then it does not satisfy $\phi_{i, 2}$ and hence $x$ does not satisfy $q_2$. If $x_i$ satisfies $\phi_{i, 2}$, then it does not satisfy $\phi_{i, 1}$, and hence $x$ does not satisfy $q_1$. In all cases, $q_1$ and $q_2$ are disjoint.\qed

\descr{Proof of Proposition~\ref{prop:pairwise}.}
Part (1) follows from the properties of set intersection. If the intersection of $|Q|$ sets is non-empty, then the intersection of each pair of sets is necessarily non-empty. 
For part (2), we present a counterexample. Consider a domain $\mathbb{D}$ with just three rows $\{x_1, x_2, x_3\}$. Let $Q \coloneqq \{q_1, q_2, q_3\}$, which are defined such that the coverage of $q_1$ is $\{x_2, x_3\}$, that of $q_2$ is $\{x_1, x_3\}$, and that of $q_3$ is $\{x_1, x_2\}$. Then the three pairwise overlap, yet $C_Q(\mathbb{D}) = \emptyset$. \qed

\section{Further Results on Random Synthetic Census Queries}
\label{app:mo-cliq-util-gain}
\iffull
Figures \ref{fig:budgetdistributionsclique} and \ref{fig:budgetdistributionsoverlap} show the utility gain via the clique number and maximum overlap algorithms on the random synthetic census queries of Section~\ref{subsec:randomcensusqueries}. These results are obtained by considering non-uniform distributions for the four quantities controlling the random generation of queries identified in Section~\ref{subsec:scalingexperiments}. Since the two algorithms frequently time-out after a small number of queries under non-uniform distribution of these quantities (see Section~\ref{subsec:scalingexperiments}), we only show the results for up to 80 queries for the maximum clique algorithm and up to 50--60 queries for the maximum overlap algorithm. Both algorithms show similar trends. 

\begin{figure*}[!ht]
     \centering
     \begin{subfigure}[b]{0.45\textwidth}
         \centering
         \includegraphics[width=\textwidth]{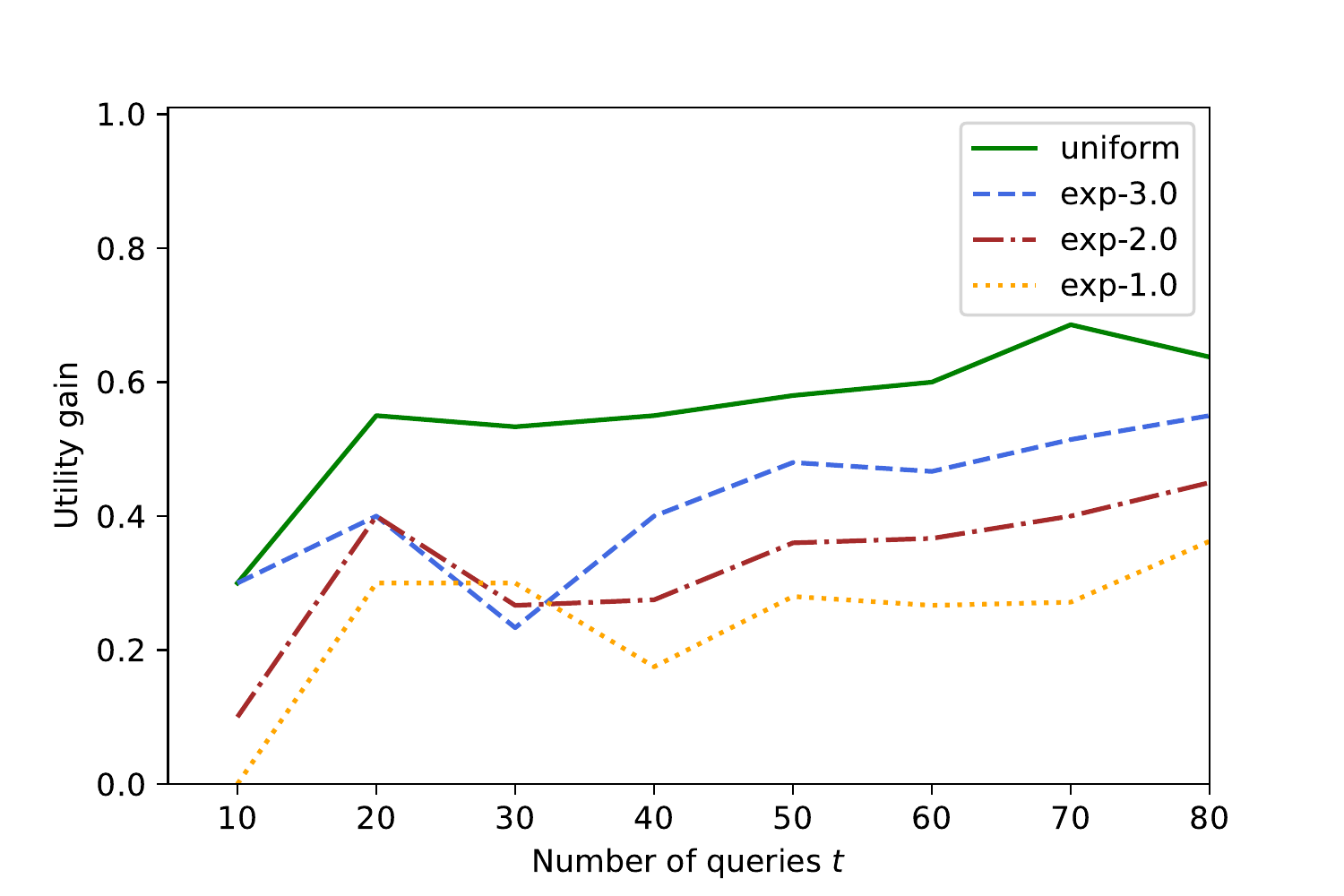}
         \caption{Influence of distribution on the number of predicates}
         \label{fig:cliquenumpredicates}
     \end{subfigure}
     \hfill
     \begin{subfigure}[b]{0.45\textwidth}
         \centering
         \includegraphics[width=\textwidth]{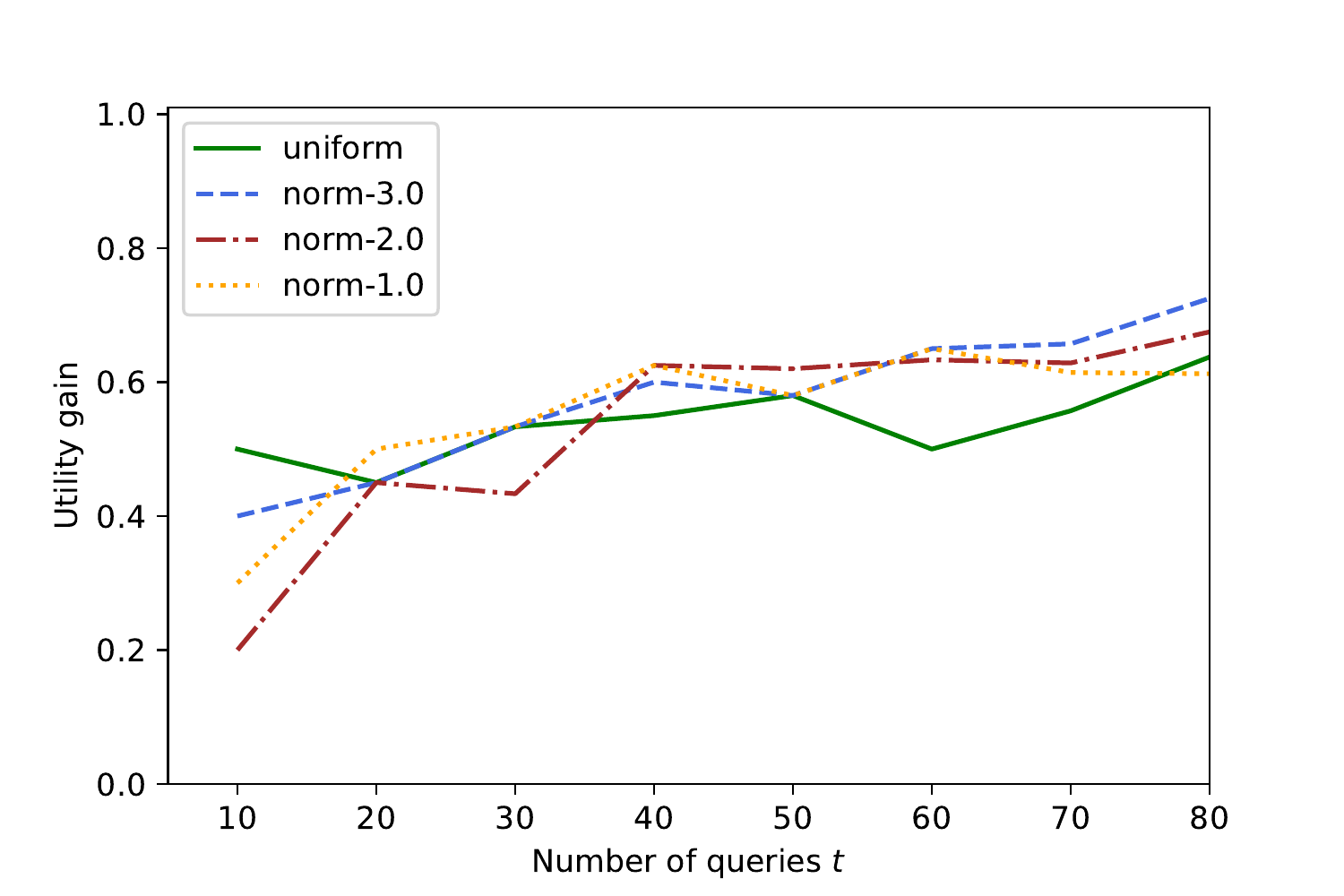}
         \caption{Influence of distribution on the attributes}
         \label{fig:cliquefields}
     \end{subfigure}
     \hfill
     \begin{subfigure}[b]{0.45\textwidth}
         \centering
         \includegraphics[width=\textwidth]{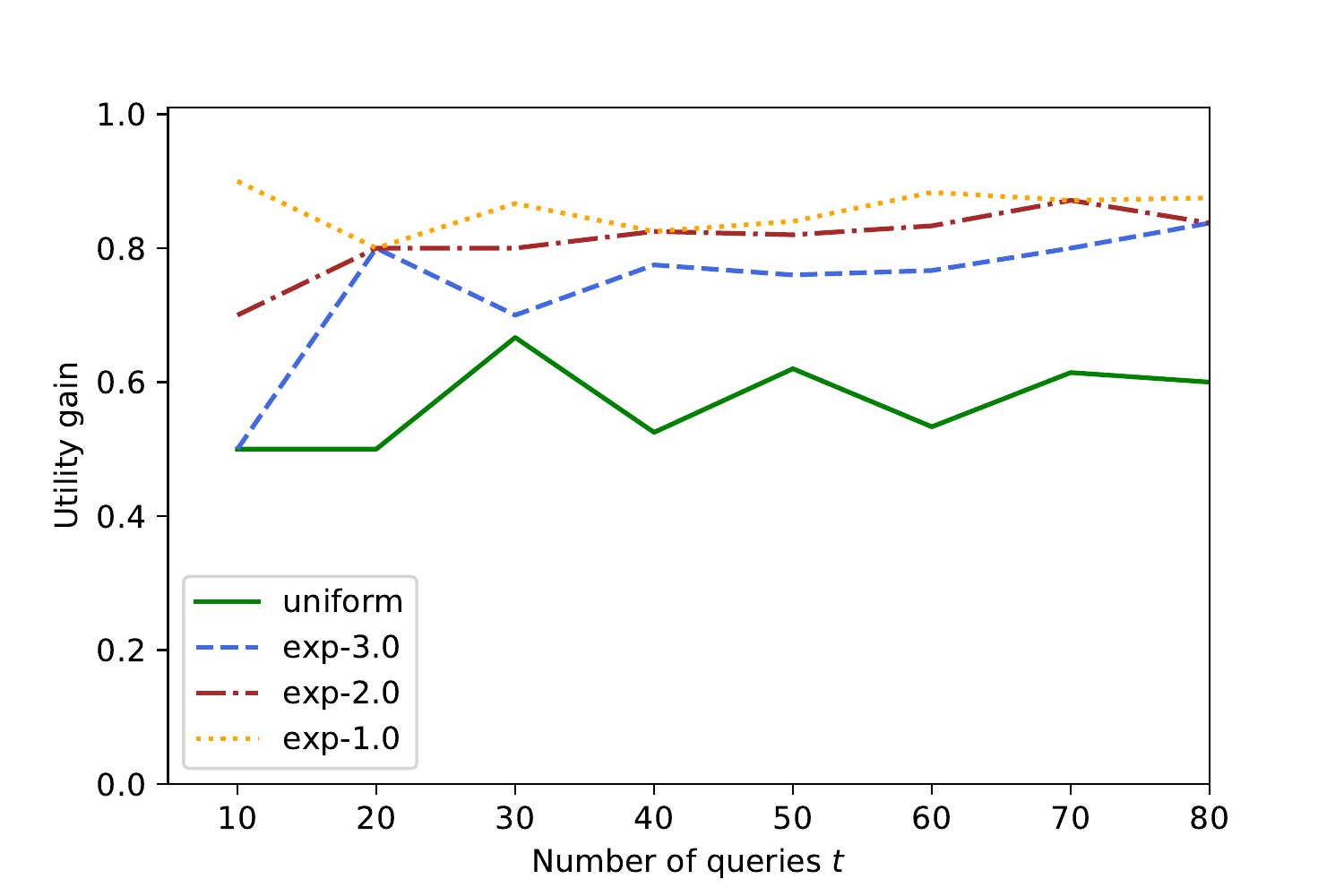}
         \caption{Influence of distribution on the number of values per predicate}
         \label{fig:cliquenumnumvalues}
     \end{subfigure}
     \hfill
     \begin{subfigure}[b]{0.45\textwidth}
         \centering
         \includegraphics[width=\textwidth]{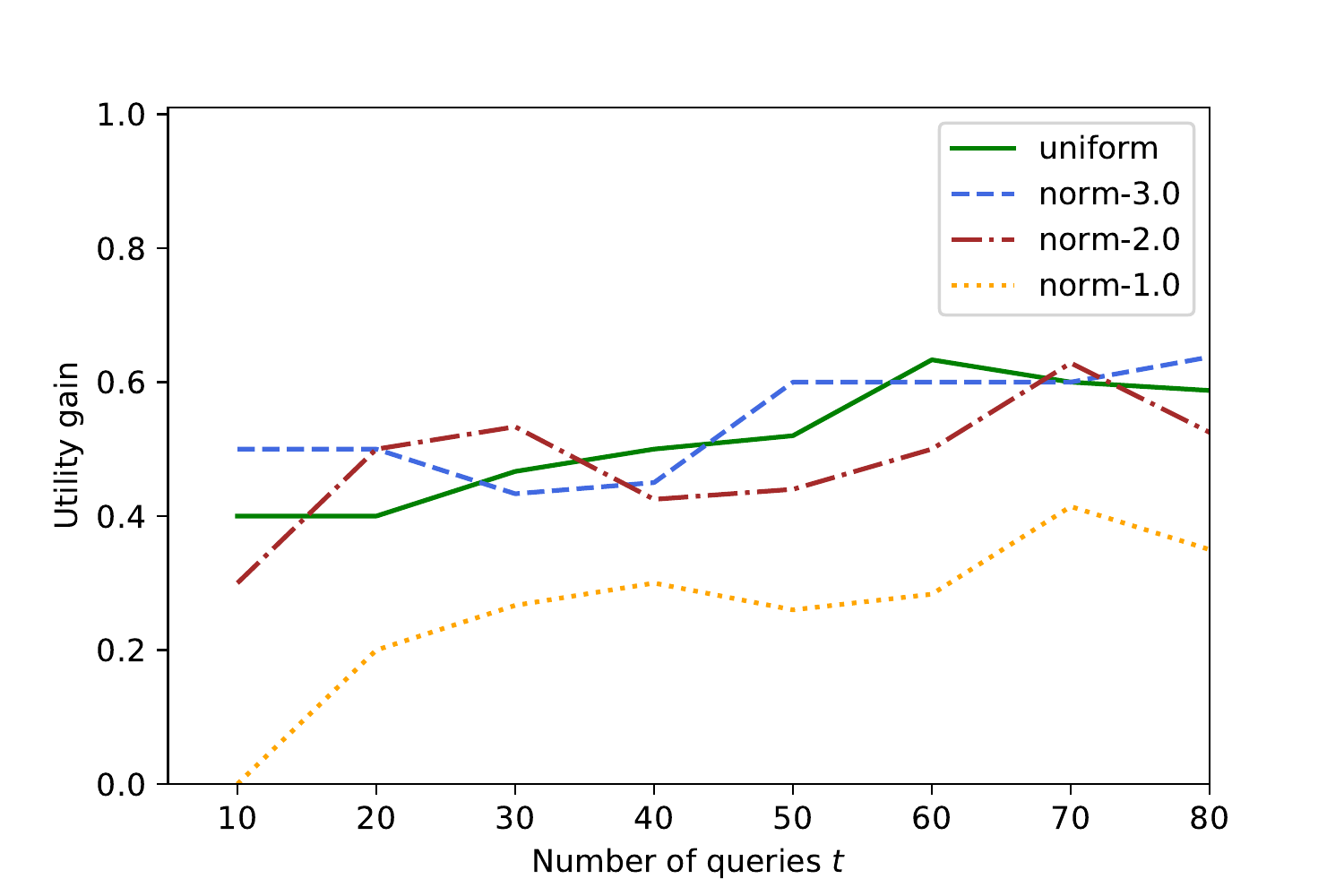}
         \caption{Influence of distribution on the attribute values}
         \label{fig:cliquenumvalues}
     \end{subfigure}
        \caption{Utility gain versus number of queries, for the \textbf{clique number} algorithm, for different distributions on the four quantities (random variables) used in generating queries from Section~\ref{subsec:scalingexperiments}. The label `uniform' refers to the uniform distribution; `exp' refers to the exponential distribution (with three different scale parameters); and `norm' refers to the normal distribution (with three different standard deviations).}
        \label{fig:budgetdistributionsclique}
\end{figure*}

\begin{figure*}[!ht]
     \centering
     \begin{subfigure}[b]{0.45\textwidth}
         \centering
         \includegraphics[width=\textwidth]{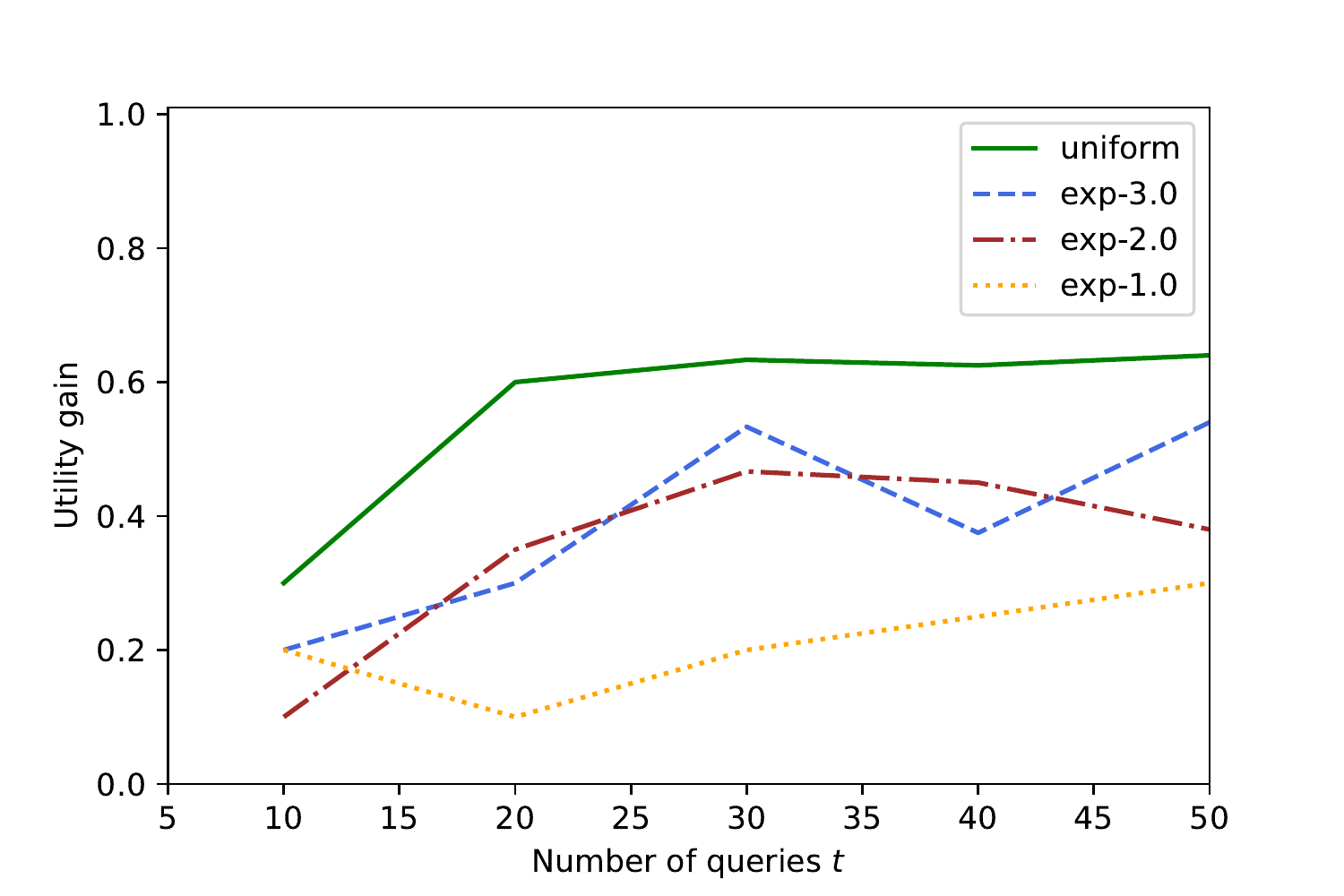}
         \caption{Influence of distribution on the number of predicates}
         \label{fig:overlapnumpredicates}
     \end{subfigure}
     \hfill
     \begin{subfigure}[b]{0.45\textwidth}
         \centering
         \includegraphics[width=\textwidth]{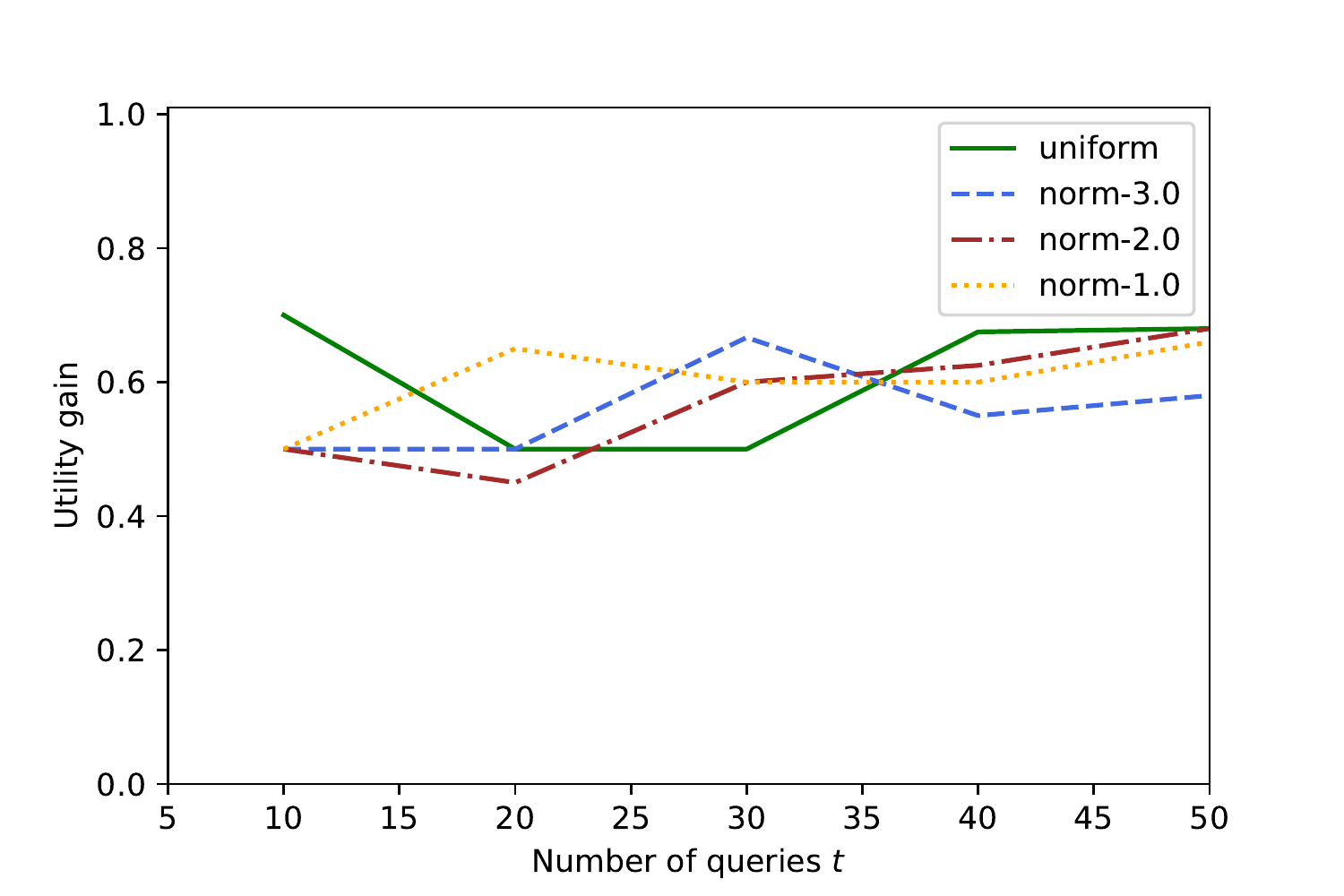}
         \caption{Influence of distribution on the attributes}
         \label{fig:overlapfields}
     \end{subfigure}
     \hfill
     \begin{subfigure}[b]{0.45\textwidth}
         \centering
         \includegraphics[width=\textwidth]{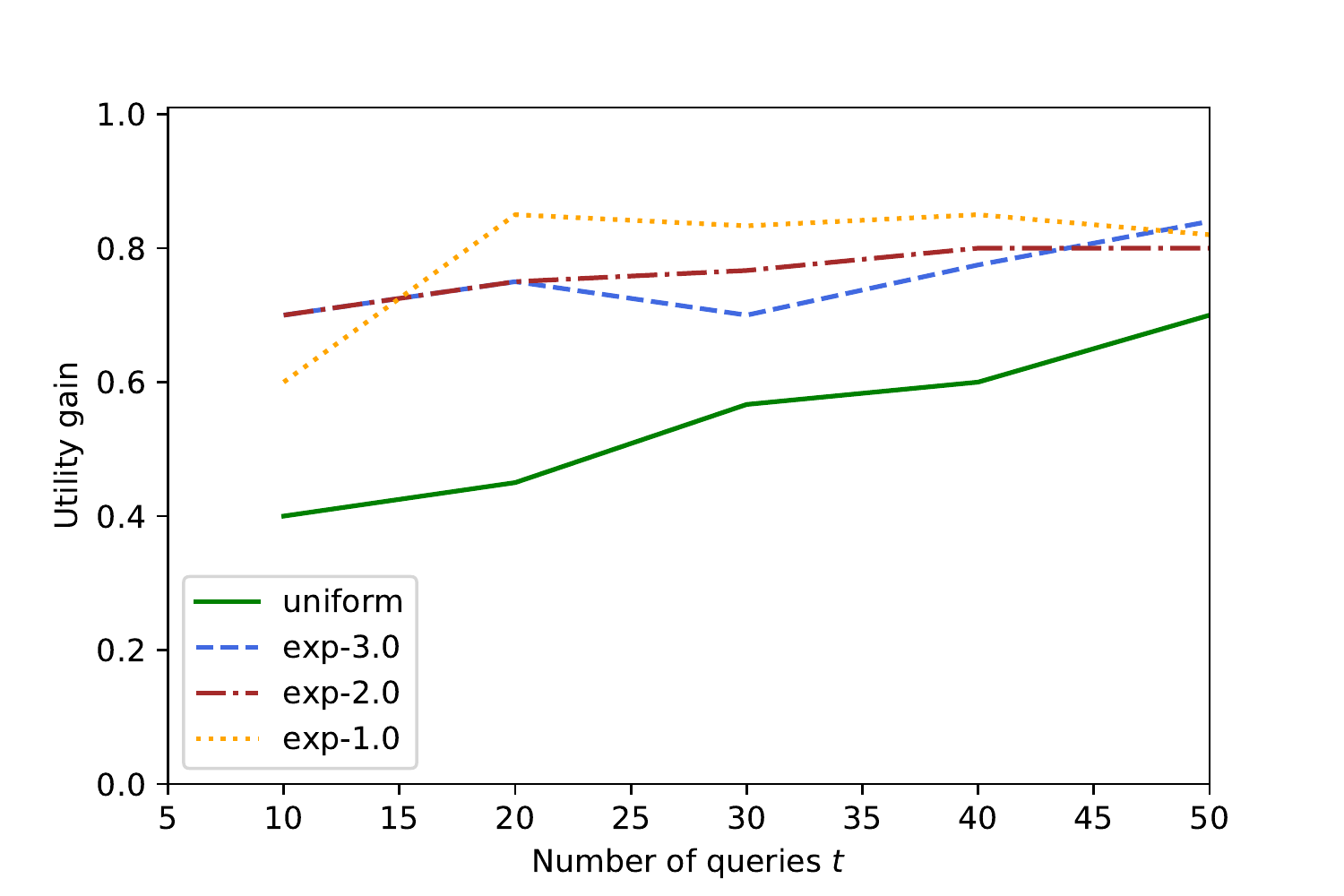}
         \caption{Influence of distribution on the number of values per predicate}
         \label{fig:overlapnumnumvalues}
     \end{subfigure}
     \hfill
     \begin{subfigure}[b]{0.45\textwidth}
         \centering
         \includegraphics[width=\textwidth]{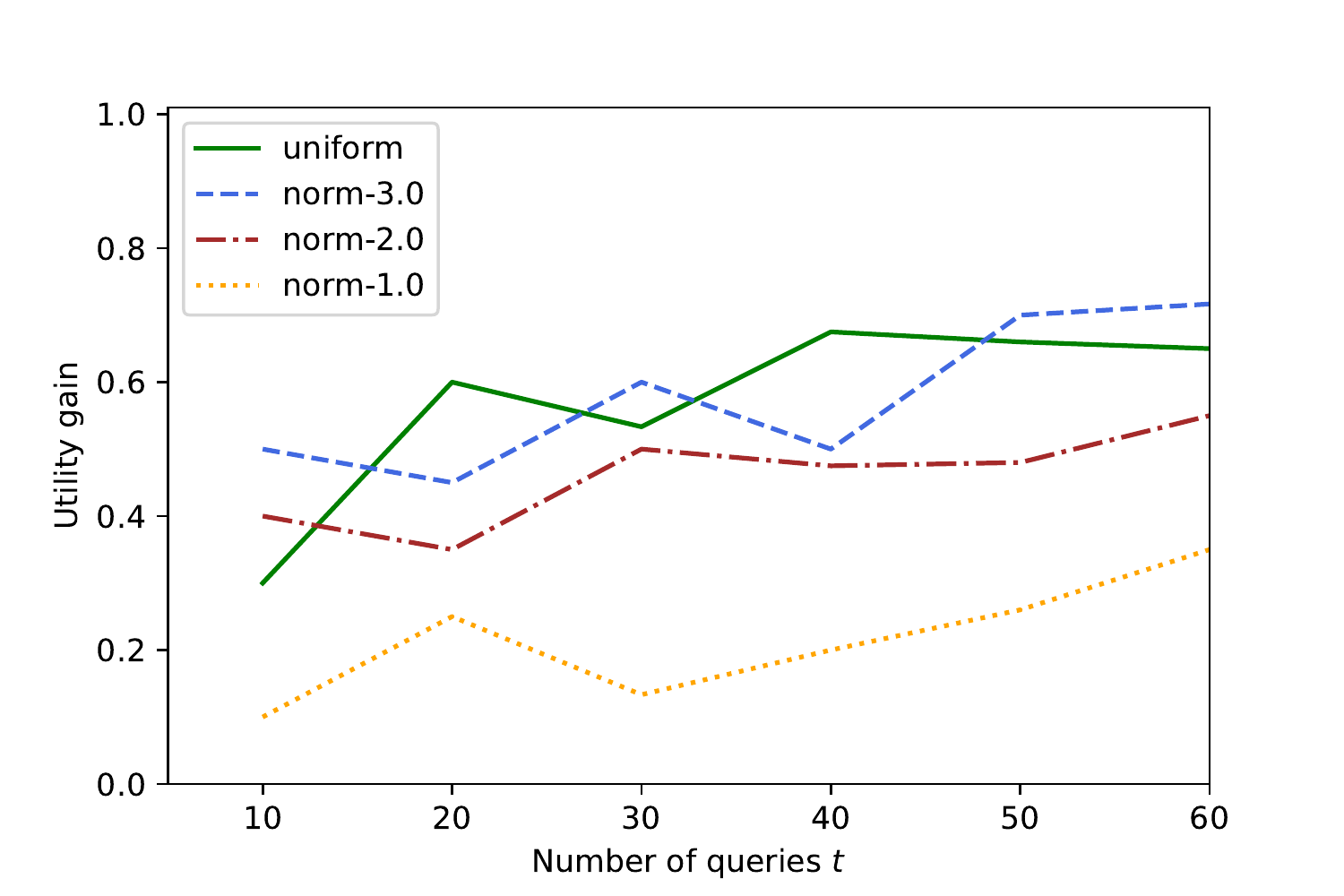}
         \caption{Influence of distribution on the attribute values}
         \label{fig:overlapnumvalues}
     \end{subfigure}
        \caption{Utility gain versus number of queries, for the \textbf{maximum overlap} algorithm, for different distributions on the four quantities (random variables) used in generating queries from Section~\ref{subsec:scalingexperiments}. The label `uniform' refers to the uniform distribution; `exp' refers to the exponential distribution (with three different scale parameters); and `norm' refers to the normal distribution (with three different standard deviations).}
        \label{fig:budgetdistributionsoverlap}
\end{figure*}
\fi

\end{document}